\documentclass[envcountsame]{llncs}
\usepackage[colorlinks=true, citecolor=blue,linkcolor=red,urlcolor=black]{hyperref}
\usepackage[utf8]{inputenc}
\usepackage{fullpage}
\usepackage{hhline}

\usepackage{xcolor,tikz,pgf,pgffor}
\usetikzlibrary{automata,shapes,backgrounds,decorations.pathreplacing}

\usepackage{amsmath}
\usepackage{amssymb,stmaryrd}
\usepackage{enumitem}
\usepackage{multirow}
\usepackage{wrapfig}

\makeatletter
\let\c@definition\c@theorem
\let\c@lemma\c@theorem
\let\c@corollary\c@theorem
\let\c@remark\c@theorem
\let\c@example\c@theorem
\let\c@proposition\c@theorem
\let\c@problem\c@theorem
\makeatother

\SetSymbolFont{stmry}{bold}{U}{stmry}{m}{n}

\title{Parameterized complexity of games with monotonically ordered \texorpdfstring{$\omega$}{w}-regular objectives
\thanks{The three authors are supported by COST Action GAMENET CA 16228. V\'eronique Bruy\`ere and Jean-Fran\c cois Raskin are both supported by the FNRS PDR project ``Subgame perfection in graph games'' (T.0088.18). Quentin Hautem is supported by a FRIA fellowship (FNRS). Jean-Fran\c cois Raskin is supported by the ERC Starting Grant inVEST (279499), by the ARC project ``Non-Zero Sum Game Graphs: Applications to Reactive Synthesis and Beyond'' (F\'ed\'eration Wallonie-Bruxelles), and by the EOS project ``Verifying Learning Artificial Intelligence Systems'' (FNRS-FWO), and he is Professeur Francqui de Recherche funded by the Francqui foundation.}
}
\titlerunning{Parameterized complexity of games with monotonically ordered $\omega$-regular objectives}

\author{V\'eronique Bruy\`ere$^1$ \and Quentin Hautem$^1$ \and Jean-Fran\c{c}ois Raskin$^2$}	 
\institute{$^1$ D\'{e}partement d'informatique, Universit\'{e} de Mons (UMONS), Mons, Belgium \\
$^2$ D\'{e}partement d'informatique, Universit\'{e} libre de Bruxelles (U.L.B.),\\ Brussels, Belgium}

\newcommand{\N}{\mathbb{N}}

\newcommand{\C}{\mathbb{C}}

\newcommand{\lexi}{\precsim}
\newcommand{\s}{s(n)}
\newcommand{\sprim}{s'(n)}

\newcommand{\playerOne}{\ensuremath{\mathcal{P}_1} } 
\newcommand{\playerTwo}{\ensuremath{\mathcal{P}_2 } }
\newcommand{\playerI}{\ensuremath{\mathcal{P}_i} }
\newcommand{\Plays}{\mathsf{Plays}}

\newcommand{\Out}{\mathsf{Out}}
\newcommand{\Obj}{\Omega}  

\newcommand{\Occ}{\mathsf{Occ}}
\newcommand{\Occinf}{\mathsf{Inf}}

\newcommand{\payoff}{\mathsf{Payoff}}

\newcommand{\Reach}{\mathsf{Reach}}
\newcommand{\Safe}{\mathsf{Safe}}
\newcommand{\Buchi}{\mathsf{Buchi}}
\newcommand{\CoBuchi}{\mathsf{CoBuchi}}
\newcommand{\GenReach}{\mathsf{GenReach}}
\newcommand{\GenBuchi}{\mathsf{GenBuchi}}
\newcommand{\Par}{\mathsf{Parity}}
\newcommand{\Rabin}{\mathsf{Rabin}}
\newcommand{\Streett}{\mathsf{Streett}}
\newcommand{\Muller}{\mathsf{Muller}}
\newcommand{\EMuller}{\mathsf{ExplMuller}}

\newcommand{\BoolBuchi}{\mathsf{BooleanBuchi}}
\newcommand{\UIBuchi}{\mathsf{UIBuchi}}
\newcommand{\UIReach}{\mathsf{UIReach}}
\newcommand{\UISafe}{\mathsf{UISafe}}
\newcommand{\unioninter}{UI}

\newcommand{\tproblem}{threshold problem}
\newcommand{\FPT}{\mathsf{FPT}}
\newcommand{\last}{\mathsf{Last}_1}
\newcommand{\M}{\mathsf{M}}
\newcommand{\Set}{S}

\begin{document}
  \maketitle

\begin{abstract} 
In recent years, two-player zero-sum games with multiple objectives have received a lot of interest as a model for the synthesis of complex reactive systems. In this framework, Player~1 wins if he can ensure that all objectives are satisfied against any behavior of Player~2. When this is not possible to satisfy all the objectives at once, an alternative is to use some preorder on the objectives according to which subset of objectives Player~1 wants to satisfy. For example, it is often natural to provide more significance to one objective over another, a situation that can be modelled with lexicographically ordered objectives for instance. Inspired by recent work on concurrent games with multiple $\omega$-regular objectives by Bouyer et al., we investigate in detail turned-based games with monotonically ordered and $\omega$-regular objectives. We study the threshold problem which asks whether player~1 can ensure a payoff greater than or equal to a given threshold w.r.t. a given monotonic preorder. As the number of objectives is usually much smaller than the size of the game graph, we provide a parametric complexity analysis and we show that our threshold problem is in $\FPT$ for all monotonic preorders and all classical types of $\omega$-regular objectives. We also provide polynomial time algorithms for B\"uchi, coB\"uchi and explicit Muller objectives for a large subclass of monotonic preorders that includes among others the lexicographic preorder. In the particular case of lexicographic preorder, we also study the complexity of computing the values and the memory requirements of optimal strategies.
\end{abstract}

\section{Introduction}
\label{sec:intro}

Two-player zero-sum games played on directed graphs form an adequate framework for the \emph{synthesis of reactive systems} facing an uncontrollable environment~\cite{PR89}. To model properties to be enforced by the reactive system within its environment, games with Boolean objectives and games with quantitative objectives have been studied, for example games with $\omega$-regular objectives~\cite{2001automata} and mean-payoff games~\cite{ZP96}. 

Recently, games with \emph{multiple} objectives have received a lot of attention since in practice, a system must usually satisfy several properties. In this framework, the system wins if it can ensure that {\em all} objectives are satisfied no matter how the environment behaves. For instance, generalized parity games are studied in~\cite{ChatterjeeHP07}, multi-mean-payoff games in~\cite{VelnerC0HRR15}, and multidimensional games with heterogeneous $\omega$-regular objectives in~\cite{BruyereHR16}. 

When multiple objectives are conflicting or if there does not exist a strategy that can enforce all of them at the same time, it is natural to consider trade-offs.  A general framework for defining trade-offs between $n$ (Boolean) objectives $ \Obj_1, \dots, \Obj_n$ consists in assigning to each infinite path $\pi$ of the game a payoff $v \in \{0,1\}^n$ such that $v(i)=1$ iff $\pi$ satisfies $\Obj_i$, and then to equip $\{0,1\}^n$ with a preorder $\precsim$ to define a preference between pairs of payoffs: $v \precsim v'$ whenever payoff $v'$ is preferred to payoff $v$. Because the ideal situation would be to satisfy {\em all} the objectives together, it is natural to assume that the preorder $\precsim$ has the following \emph{monotonicity} property: if $v'$ is such that whenever $v(i)=1$ then $v'(i)=1$, then it should be the case that $v'$ is preferred to $v$. 

As an illustration, let us consider a game in which Player~1 strives to enforce three objectives: $\Omega_1$, $\Omega_2$, and $\Omega_3$. Assume also that Player~$1$ has no strategy ensuring all three objectives at the same time, that is, Player~1 cannot ensure the objective $\Omega_1 \cap \Omega_2 \cap \Omega_3$. Then several options can be considered, see e.g.~\cite{BouyerBMU12}. First, we could be interested in a strategy of Player~1 ensuring a maximal subset of the three objectives. Indeed, a strategy that enforces both $\Omega_1$ and $\Omega_3$ should be preferred to a strategy that enforces $\Omega_3$ only. This preference is usually called the \emph{subset preorder}. Now, if $\Omega_1$ is considered more important than $\Omega_2$ itself considered more important than $\Omega_3$, then a strategy that ensures the most important possible objective should be considered as the most desirable. This preference is called the {\em maximize preorder}. Finally, we could also translate the relative importance of the different objectives into a {\em lexicographic preorder} on the payoffs:  satisfying $\Omega_1$ and $\Omega_2$ would be considered as more desirable than satisfying $\Omega_1$ and $\Omega_3$ but not $\Omega_2$. Those three examples are all monotonic preorders.

In this paper, we consider the following threshold problem: given a game graph $G$, a set of \emph{$\omega$-regular objectives}\footnote{We cover all classical $\omega$-regular objectives: reachability, safety, B\"uchi, co-B\"uchi, parity, Rabin, Streett, explicit Muller, or Muller.} $\Obj_1, \dots, \Obj_n$, a monotonic preorder $\precsim$ on the set $\{0,1\}^n$ of payoffs, and a threshold $\mu$, decide whether Player~1 has a strategy such that for all strategies of Player 2, the outcome of the game has payoff $v$ greater than or equal to $\mu$ (for the specified preorder), i.e. $\mu \precsim v$. As the number $n$ of objectives is typically much smaller than the size of the game graph $G$, it is natural to consider a parametric analysis of the complexity of the threshold problem in which the number of objectives and their size are considered to be fixed parameters of the problem. Our main results are as follows.

\medskip\noindent{\bf Contributions.~}
First, we provide {\em fixed parameter tractable solutions} to the threshold problem for {\em all} monotonic preorders and for {\em all} classical types of $\omega$-regular objectives. Our solutions rely on the following ingredients:
\begin{enumerate}
\item We show that solving the threshold problem is equivalent to \emph{solve a game with a single objective} $\Obj$ that is a union of intersections of objectives taken among $\Obj_1, \dots, \Obj_n$ (Theorem~\ref{thm:reduction}). This is possible by \emph{embedding} the monotonic preorder $\precsim$ in the subset preorder and by translating the threshold $\mu$ in preorder $\precsim$ into an antichain of thresholds in the subset preorder. A threshold in the subset preorder is naturally associated with a conjunction of objectives, and an antichain of thresholds leads to a union of such conjunctions.
\item We provide a fixed parameter tractable algorithm to solve games with a single objective $\Obj$ as described previously for all types of $\omega$-regular objectives $\Obj_1, \dots, \Obj_n$, leading to a \emph{fixed parameter algorithm for the threshold problem} (Theorem~\ref{thm:lexiFPT}). Those results build on the recent breakthrough of Calude et al. that provides a quasipolynomial time algorithm for parity games as well as their fixed parameter tractability~\cite{Calude}, and on the fixed parameter tractability of games with an objective defined by a \emph{Boolean combination of B\"uchi objectives} (Proposition~\ref{prop:FPT}).
\end{enumerate}

\noindent
Second, we consider games with a preorder $\precsim$ having a \emph{compact embedding}, with the main condition that the antichain of thresholds resulting from the embedding in the subset preorder is of {\em polynomial size}. The maximize preorder, the subset preorder, and the lexicographic preorder, given as examples above, all possess this property.
For games with a compact embedding, we go \emph{beyond fixed parameter tractability} as we are able to provide deterministic polynomial time solutions for B\"uchi, coB\"uchi, and explicit Muller objectives (Theorem~\ref{thm:poly}). Polynomial time solutions are not possible for the other types of $\omega$-regular objectives as we show that the threshold problem for the \emph{lexicographic preorder} with reachability, safety, parity, Rabin, Streett, and Muller objectives cannot be solved in polynomial time unless $\mathsf{P}=\mathsf{PSPACE}$ (Theorem~\ref{thm:homogeneoushyp}). 
Finally, we present a \emph{full picture} of the study of the lexicographic preorder for each studied objective. We give the exact complexity class of the \tproblem, show that we can obtain the values from the \tproblem\ (which thus yields a polynomial algorithm for B\"uchi, co-B\"uchi and Explicit Muller objectives, and an $\mathsf{FPT}$ algorithm for the other objectives) and provide tight memory requirements for the optimal and winning strategies (Table~\ref{table:homogeneous}).

\medskip\noindent{\bf Related work.~}
In~\cite{BouyerBMU12}, Bouyer et al. investigate concurrent games with multiple objectives leading to payoffs in $\{0,1\}^n$ which are ordered using Boolean circuits. While their threshold problem is slightly more general than ours, their games being concurrent and their preorders being not necessarily monotonic, the algorithms that they provide are nondeterministic and guess witnesses whose size depends polynomially not only in the number of objectives but also in the size of the game graph. Their algorithms are sufficient to establish membership to {$\mathsf{PSPACE}$} for all classical types of $\omega$-regular objectives but they do not provide a basis for the parametric complexity analysis of the threshold problem. In stark contrast, we provide deterministic algorithms whose complexity only depends polynomially in the size of the game graph. Our new deterministic algorithms are thus instrumental to a finer complexity analysis that leads to fixed parameter tractability for all monotonic preorders and all $\omega$-regular objectives.  We also provide tighter lower-bounds for the important special case of lexicographic preorder, in particular for parity objectives.

The particular class of games with multiple B\"uchi objectives ordered with the maximize preorder has been considered in~\cite{AlurKW08}. The interested reader will find in that paper clear practical motivations for considering multiple objectives and ordering them. The lexicographic ordering of objectives has also been considered in the context of quantitative games: lexicographic mean-payoff games in~\cite{BloemCHJ09}, some special cases of lexicographic quantitative games in~\cite{BruyereMR14,0001MPRW17}, and lexicographically ordered energy objectives in~\cite{ColcombetJLS17}.

In~\cite{AlmagorK17} and~\cite{KupfermanPV14}, the authors investigate partially (or totally) ordered specifications expressed in LTL. None of their complexity results leads to the results of this paper since the complexity is de facto much higher with objectives expressed in LTL. Moreover no $\mathsf{FPT}$ result is provided in those references.

\medskip\noindent{\bf Structure of the paper.~} In Section~\ref{sec:preliminaries}, we present all the useful notions about games with monotonically ordered $\omega$-regular objectives. In Section~\ref{sec:FPT}, we show that solving the threshold problem is equivalent to solve a game with a single objective that is a union of intersections of objectives (Theorem~\ref{thm:reduction}), and we establish the main result of this paper: the fixed parameter complexity of the threshold problem (Theorem~\ref{thm:lexiFPT}). Section~\ref{sec:homo} is devoted to games with a compact embedding and in particular to the threshold problem for lexicographic games. The last section is dedicated to the study of computing the values and memory requirements of optimal strategies in the case of lexicographic games (Table~\ref{table:homogeneous}).


\section{Preliminaries} \label{sec:preliminaries}

We consider zero-sum turn-based games played by two players, $\playerOne$ and $\playerTwo$, on a finite directed graph. Given \emph{several objectives}, we associate with each play of this game a vector of bits called \emph{payoff}, the components of which indicate the objectives that are satisfied. The set of all payoffs being equipped with a \emph{preorder}, $\playerOne$ wants to ensure a payoff greater than or equal to a given threshold against any behavior of $\playerTwo$. In this section we give all the useful notions and the studied problem.

\paragraph{\bf Preorders. ~} Given some non-empty set $P$, a \emph{preorder} over $P$ is a binary relation $\precsim$ $\subseteq P \times P$ that is reflexive and transitive. The \emph{equivalence relation} $\sim$ associated with $\precsim$ is defined such that $x \sim y$ if and only if $x \precsim y$ and  $y \precsim x$. The \emph{strict partial order} $\prec$ associated with $\precsim$ is then defined such that $x \prec y$ if and only if $x \precsim y$ and $x \not\sim y$. A preorder $\precsim$ is \emph{total} if $x \precsim y$ or $y \precsim x$ for all $x, y \in P$. A set $S \subseteq P$ is \emph{upper-closed} if for all $x \in S$, $y \in P$, if $x \precsim y$, then $y \in S$. An \emph{antichain} is a set $S \subseteq P$ of pairwise incomparable elements, that is, for all $x, y \in S$, if $x \neq y$, then $x \not \precsim y$ and $y \not \precsim x$.

\paragraph{\bf Game structures.~} We give below the definition of a game structure and notations on plays.
\begin{definition}
A \emph{game structure} is a tuple $G = (V_1,V_2,E)$ where
\begin{itemize}
\item $(V,E)$ is a finite directed graph, with $V = V_1 \cup V_2$ the set of vertices and $E \subseteq V \times V$ the set of edges such that\footnote{This condition guarantees that there is no deadlock. It can be assumed w.l.o.g. for all the problems considered in this article.} for each $v \in V$, there exists $(v,v') \in E$ for some $v' \in V$,
\item $(V_1,V_2)$ forms a partition of $V$ such that $V_i$ is the set of vertices controlled by player $\playerI$ with $i \in \{1,2\}$.
\end{itemize}
\end{definition}

A \emph{play} of $G$ is an infinite sequence of vertices $\pi = v_0 v_1 \ldots \in V^{\omega}$ such that $(v_k,v_{k+1}) \in E$ for all $k \in \N$. We denote by $\Plays(G)$ the set of plays in $G$. \emph{Histories} of $G$ are finite sequences $\rho = v_0 \ldots v_k \in V^+$ defined in the same way. 
Given a play $\pi = v_0 v_1 \ldots$, the set $\Occ(\pi)$ denotes the set of vertices that occur in $\pi$, and the set $\Occinf(\pi)$ denotes the set of vertices visited infinitely often along $\pi$, i.e., $\Occ(\pi) = \{v \in V \mid \exists k \geq 0, v_k = v \}$ and $\Occinf(\pi) = \{v \in V \mid \forall k \geq 0, \exists l \geq k,\ v_l = v\}$. Given a set $U \subseteq V$ and a set $\Obj \subseteq V^{\omega}$, we denote by $U^c$ the set $V \setminus U$ and by $\overline{\Obj}$ the set $V^{\omega} \setminus \Obj$.

\paragraph{\bf Strategies.~} 

A \emph{strategy} $\sigma_i$ for $\playerI$ is a function $\sigma_i\colon V^*V_i \rightarrow V$ assigning to each history $\rho v \in V^*V_i$ a vertex $v' = \sigma_i(\rho v)$ such that $(v,v') \in E$. It is \emph{memoryless} if $\sigma_i(\rho v) = \sigma_i(\rho'v)$ for all histories $\rho v, \rho'v$ ending with the same vertex $v$, that is, if $\sigma_i$ is a function $\sigma_i\colon V_i \rightarrow V$. It is \emph{finite-memory} if it can be encoded by a deterministic \emph{Moore machine} ${\cal M} = (M, m_0, \alpha_u, \alpha_n)$ where $M$ is a finite set of states (the memory of the strategy), $m_0 \in M$ is the initial memory state, $\alpha_u\colon M \times V \rightarrow M$ is the update function, and $\alpha_n\colon M \times V_i \rightarrow V$ is the next-action function. The Moore machine $\cal M$ defines a strategy $\sigma_i$ such that $\sigma_i(\rho v) = \alpha_n(\widehat{\alpha}_u(m_0,\rho),v)$ for all histories $\rho v \in V^*V_i$, where $\widehat{\alpha}_u$ extends $\alpha_u$ to histories as expected. The \emph{size} of the strategy $\sigma_i$ is the size $|M|$ of its machine $\cal M$. Note that $\sigma_i$ is memoryless when $|M| = 1$.

The set of all strategies of $\playerI$ is denoted by $\Sigma_i$. 
Given a strategy $\sigma_i$ of $\playerI$, a play $\pi = v_0 v_1 \ldots$ of $G$ is \emph{consistent} with $\sigma_i$ if $v_{k+1} = \sigma_i(v_0 \ldots v_k)$ for all $k \in \N$ such that $v_k \in V_i$. Consistency is naturally extended to histories in a similar fashion. 
Given an \emph{initial vertex} $v_0$, and a strategy $\sigma_i$ of each player $\playerI$, we have a unique play consistent with both strategies $\sigma_1, \sigma_2$, called \emph{outcome} and denoted by $\Out(v_0,\sigma_1,\sigma_2)$.

\paragraph{\bf Single objectives and ordered objectives.~}

An \emph{objective for $\playerOne$} is a set of plays $\Obj \subseteq \Plays(G)$. A \emph{game} $(G,\Obj)$ is composed of a game structure $G$ and an objective~$\Obj$. A play $\pi$ is \emph{winning} for $\playerOne$ if $\pi \in \Obj$, and losing otherwise. As the studied games are zero-sum, $\playerTwo$ has the opposite objective $\overline{\Obj}$, meaning that a play $\pi$ is winning for $\playerOne$ if and only if it is losing for $\playerTwo$. Given a game $(G,\Obj)$ and an initial vertex $v_0$, a strategy $\sigma_1$ for $\playerOne$ is \emph{winning from} $v_0$ if $\Out(v_0,\sigma_1,\sigma_2) \in \Obj$ for all strategies $\sigma_2$ of $\playerTwo$. Vertex $v_0$ is thus called \emph{winning} for $\playerOne$. We also say that $\playerOne$ is winning from $v_0$ or that he can \emph{ensure} $\Obj$ \emph{from} $v_0$. Similarly the winning vertices of $\playerTwo$ are those from which $\playerTwo$ can ensure his objective $\overline{\Obj}$. 

A game $(G,\Obj)$ is \emph{determined} if each of its vertices is either winning for $\playerOne$ or winning for $\playerTwo$. Martin's theorem~\cite{Martin75} states that all games with Borel objectives are determined. The problem of \emph{solving a game} $(G,\Obj)$ means to decide, given an initial vertex $v_0$, whether $\playerOne$ is winning from $v_0$ (or dually whether $\playerTwo$ is winning from $v_0$ when the game is determined).  

Instead of a \emph{single} objective $\Obj$, one can consider \emph{several} objectives $\Obj_1,\ldots,\Obj_n$ that are \emph{ordered} with respect to a preorder $\precsim$ over $\{0,1\}^n$ in the following way. We first define the payoff of a play as a vector\footnote{Note that in the sequel, we often manipulate equivalently vectors in $\{0,1\}^n$ and sequences of $n$ bits.} of bits the components of which indicate the objectives that are satisfied. 

\begin{definition}
Given a game structure $G = (V_1,V_2,E)$, and $n$ objectives $\Obj_1,\ldots,\Obj_n \subseteq \Plays(G)$, the \emph{payoff} function $\mathsf{Payoff} \colon \Plays(G) \rightarrow \{0,1\}^n$ assigns a vector of bits to each play $\pi \in \Plays(G)$, where for all $k \in \{1, \ldots, n\}$, $\payoff_k(\pi) = 1$ if $\pi \in \Obj_k$ and $0$ otherwise.
\end{definition}

Given the preorder $\precsim$ over $\{0,1\}^n$, $\playerOne$ prefers a play $\pi$ to a play $\pi'$ whenever $\mathsf{Payoff}(\pi') \precsim \mathsf{Payoff}(\pi)$. We call \emph{ordered game} the tuple $(G,\Obj_1,\ldots, \Obj_n,\precsim)$, the payoff function of which is defined w.r.t. the objectives $\Obj_1, \ldots, \Obj_n$ and its values are ordered with $\precsim$.
In this context, we are interested in the following problem.

\begin{problem} \label{prob:threshold}
The \emph{threshold problem} for ordered games $(G,\Obj_1,\ldots, \Obj_n,\precsim)$ asks, given a threshold $\mu \in \{0,1\}^n$ and an initial vertex $v_0 \in V$, to decide whether $\playerOne$ (resp. $\playerTwo$) has a strategy to ensure the objective $\Obj = \{\pi \in \Plays(G) \mid \payoff(\pi) \succsim \mu \}$ from $v_0$ (resp. $\overline{\Obj} = \{\pi \in \Plays(G) \mid \payoff(\pi) \not\succsim \mu \}$).\footnote{Note that when $n = 1$ and $\precsim$ is the usual order $\leq$ over $\{0,1\}$, we recover the notion of single objective with the threshold $\mu = 1$.} 
\end{problem}
In case $\playerOne$ (resp. $\playerTwo$) has such a winning strategy, we also say that he can \emph{ensure} (resp. \emph{avoid}) \emph{a payoff} $\succsim \mu$. 

Classical examples of preorders are the following ones~\cite{BouyerBMU12}. Let $x, y \in \{0,1\}^n$.
\begin{itemize}
\item \emph{Counting}: $x \precsim y$ if and only if $|\{j \mid x_j = 1\}| \leq |\{j \mid y_j = 1\}|$. The aim of $\playerOne$ is to maximize the number of satisfied objectives.
\item \emph{Subset}: $x \precsim y$ if and only if $\{j \mid x_j = 1\} \subseteq \{j \mid y_j = 1\}$. The aim of $\playerOne$ is to maximize the subset of satisfied objectives with respect to the inclusion.
\item \emph{Maximise}: $x \precsim y$ if and only if $\max \{j \mid x_j = 1\} \leq \max \{j \mid y_j = 1\}$. The aim of $\playerOne$ is to maximize the higher index of the satisfied objectives.
\item \emph{Lexicographic}: $x \precsim y$ if and only if either $x = y$ or $\exists j \in \{1,\ldots,n\}$ such that $x_j < y_j$ and $\forall k \in \{1,\ldots,j-1\}$, $x_k = y_k$. The objectives are ranked according to their importance. The aim of $\playerOne$ is to maximise the payoff with respect to the induced lexicographic order.
\end{itemize}

\begin{figure}[h]
	\centering
    \begin{tikzpicture}[scale=4]
    \everymath{\scriptstyle}
    \draw (0,0) node [rectangle, inner sep=5pt, draw] (A) {$v_0$};
    \draw (0.5,0.15) node [circle, draw] (B) {$v_1$};
    \draw (0.5,-0.15) node [circle, draw] (C) {$v_2$};
    
    \draw[->,>=latex] (A) to (B);
    
    \draw[->,>=latex] (A) to (C);
    \draw[->,>=latex] (C) to[bend left] (A);
    
    \draw[->,>=latex] (C) .. controls +(45:0.3cm) and +(315:0.3cm) .. (C);
    \draw[->,>=latex] (B) .. controls +(45:0.3cm) and +(315:0.3cm) .. (B);
	\path (0,0.15) edge [->,>=latex] (A);    
    
    \end{tikzpicture}
  \caption{A simple lexicographic game.}
  \label{fig:exemples}
\end{figure}
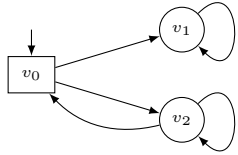

In this article, we \emph{focus on monotonic preorders}. A preorder $\precsim$ is \emph{monotonic} if it is compatible with the subset preorder, i.e. if $\{i \mid x_i = 1\} \subseteq \{i \mid y_i = 1\}$ implies $x \precsim y$. Hence a preorder is monotonic if satisfying more objectives never results in a lower payoff value. This is a \emph{natural property} shared by all the examples of preorders given previously.

\begin{example}
Consider the game structure $G$ depicted on Figure~\ref{fig:exemples}, where circle vertices belong to $\playerOne$ and square vertices belong to $\playerTwo$. We consider the ordered game $(G,\Obj_1,\Obj_2,\precsim)$ with $\Obj_i = \{\pi \in \Plays(G) \mid v_i \in \Occinf(\pi) \}$ for $i = 1,2$ and the lexicographic preorder $\precsim$. Therefore the function $\payoff$ assigns value $1$ to each play $\pi$ on the first (resp. second) bit if and only if $\pi$ visits infinitely often vertex $v_1$ (resp. $v_2$). In this ordered game, $\playerOne$ has a strategy to ensure a payoff $\succsim 01$ from $v_0$. Indeed, consider the memoryless strategy $\sigma_1$ that loops in $v_1$ and in $v_2$. Then, from $v_0$, $\playerTwo$ decides to go either to $v_1$ leading to the payoff $10$, or to $v_2$ leading to the payoff $01$. As $10 \succsim 01$, this shows that any play $\pi$ consistent with $\sigma_1$ satisfies $\payoff(\pi) \succsim 01$. Notice that while $\playerOne$ can ensure a payoff $\succsim 01$ from $v_0$, he has no strategy to enforce the single objective $\Obj_1$ and similarly no strategy to enforce $\Obj_2$.    
\end{example}

\paragraph{\bf Homogeneous $\omega$-regular objectives.~}

In the sequel of this article, given a monotonically ordered game $(G,\Obj_1,\ldots, \Obj_n,\precsim)$, we want to study the threshold problem described in Problem~\ref{prob:threshold} for \emph{homogeneous $\omega$-regular objectives}, in the sense that all the objectives $\Obj_1,\ldots,\Obj_n$ are of the same type, and taken in the following list of well-known $\omega$-regular objectives. 

\noindent Given a game structure $G = (V_1,V_2,E)$ and a subset $U$ of $V$ called \emph{target set}:
\begin{itemize}
\item The \emph{reachability objective} asks to visit a vertex of $U$ at least once, i.e. 
$\Reach(U) = \{ \pi \in \Plays(G) \mid \Occ(\pi) \cap U \ne \emptyset \}$.
\item The \emph{safety objective} asks to always stay in the set $U$, i.e. 
$\Safe(U) = \{ \pi \in \Plays(G) \mid \Occ(\pi) \cap U^c = \emptyset \}$.
\item The \emph{B\"uchi objective} asks to visit infinitely often a vertex of $U$, i.e.
$\Buchi(U) = \{ \pi \in \Plays(G) \mid \Occinf(\pi) \cap U \ne \emptyset \}$.
\item The \emph{co-B\"uchi objective} asks to eventually always stay in the set $U$, i.e.
$\CoBuchi(U) = \{\pi \in \Plays(G) \mid \Occinf(\pi) \cap U^c = \emptyset \}$.
\end{itemize}
Given a \emph{family} $\mathcal{F} = (F_i)_{i=1}^k$ of sets $F_i \subseteq V$, and a family of \emph{pairs} $((E_i,F_i)_{i=1}^k)$, with $E_i,F_i \subseteq V$:
\begin{itemize}
\item The \emph{explicit Muller objective} asks that the set of vertices seen infinitely often is exactly one among the sets of $\mathcal{F}$, i.e. $\EMuller(\mathcal{F}) = \{ \pi \in \Plays(G) \mid \exists i \in \{1,\ldots,k\}, \Occinf(\pi) = F_i \}.$
\item The \emph{Rabin objective} asks that there exists a pair $(E_i,F_i)$ such that a vertex of $F_i$ is visited infinitely often while no vertex of $E_i$ is visited infinitely often, i.e.
$\Rabin((E_i,F_i)_{i=1}^k) = \{\pi \in \Plays(G) \mid \exists i \in \{1,\ldots,k\}, \Occinf(\pi) \cap E_i = \emptyset \text{ and }  \Occinf(\pi) \cap F_i \ne \emptyset\}$.
\item The \emph{Streett objective} asks that for each pair $(E_i,F_i)$, a vertex of $E_i$ is visited infinitely often or no vertex of $F_i$ is visited infinitely often, i.e.
$\Streett((E_i,F_i)_{i=1}^k) = \{\pi \in \Plays(G) \mid \forall i \in \{1,\ldots,k\}, \Occinf(\pi) \cap E_i \ne \emptyset \text{ or } $ $ \Occinf(\pi) \cap F_i = \emptyset\}$.
\end{itemize}
Given a \emph{coloring} function $p \colon V \rightarrow \{0,\ldots,d\}$ that associates with each vertex a color, and $\mathcal{F} = (F_i)_{i=1}^k$ a family of subsets $F_i$ of $p(V)$:
\begin{itemize}
\item The \emph{parity objective} asks that the minimum color seen infinitely often is even, i.e.
$\Par(p) = \{ \pi \in \Plays(G) \mid \min_{v \in \Occinf(\pi)} p(v) \text{ is even}\}$.
\item The \emph{Muller objective} asks that the set of colors seen infinitely often is exactly one among the sets of $\mathcal{F}$, i.e.
$\Muller(p,\mathcal{F}) = \{ \pi \in \Plays(G) \mid \exists i \in \{1,\ldots,k\}, p(\Occinf(\pi)) = F_i \}.$
\end{itemize}

In the sequel, we make the \emph{assumption} that the considered preorders are monotonic, and by \emph{ordered game}, we always mean monotonically ordered games.
When the objectives of an ordered game are of kind $X$, we speak of an \emph{ordered $X$ game}, or of a \emph{$\precsim$ $X$ game} if we want to specify the used preorder $\precsim$.   
As already mentioned, when $n = 1$, an ordered game (with $\precsim$ equal to $\leq$) resumes to a game $(G,\Obj)$ with a single objective $\Obj$, that is traditionally called an $\Obj$ game. For instance, an ordered game $(G,\Obj_1,\ldots, \Obj_n,\precsim)$ where $\Obj_1, \ldots, \Obj_n$ are reachability objectives and $\precsim$ is the lexicographic preorder is called a lexicographic reachability game, and when $n = 1$ $(G,\Obj_1)$ is called a reachability game.

Note that given an ordered game with $n$ non-homogeneous $\omega$-regular objectives $\Obj_i$, we can always construct a new equivalent ordered parity game, since each objective $\Obj_i$ can be translated into a parity objective~\cite{2001automata}.

\paragraph{\bf Useful results on games with a single objective.}

Let us end this section by providing some results on games with a single $\omega$-regular objective taken among those defined previously or among the additional ones given herafter. All these results will be useful in the proofs.

Let $G$ be a game structure and $U_1, \ldots, U_m$ be $m$ target sets and $\phi$ be a Boolean formula over variables $x_1, \ldots, x_m$. We say that a play $\pi$ satisfies $(\phi,U_1, \ldots, U_m)$ if the truth assignment (~$x_i = 1$ if and only if $\Occinf(\pi) \cap U_i \neq \emptyset$, and $x_i = 0$ otherwise~) satisfies $\phi$.  
\begin{itemize}
\item \emph{Boolean combination of B\"uchi objectives}, or shortly \emph{Boolean B\"uchi} objective: 
$$\BoolBuchi(\phi,U_1,\ldots,U_m) = \{ \pi \in \Plays(G) \mid \pi \mbox{ satisfies } (\phi,U_1, \ldots, U_m) \}.$$
\end{itemize}
All operators $\vee$, $\wedge$, $\neg$ are allowed in Boolean B\"uchi objectives. However we denote by $|\phi|$ the \emph{size} of $\phi$ equal to the number of disjunctions and conjunctions inside $\phi$, and we say that the Boolean B\"uchi objective $\BoolBuchi(\phi,U_1,\ldots,U_m)$ is \emph{of size $|\phi|$ and with $m$ variables}. The definition of $|\phi|$ is not the classical one that usually counts the number of operators $\vee, \wedge, \neg$ and variables. This is not a restriction since one can transform any Boolean formula $\phi$ into one such that negations only apply on variables. 

We need to introduce some other kinds of $\omega$-regular objectives with Boolean combinations of objectives that are limited to
\begin{itemize}
\item intersections of objectives: like a \emph{generalized reachability} objective $\GenReach(U_1,\ldots,U_m)$ or a \emph{generalized B\"uchi} objective $\GenBuchi(U_1,\ldots,U_m)$,
\item unions of intersections (\unioninter) of objectives: like a \emph{\unioninter\  reachability} objective $$\UIReach(U_{1,1},\ldots,U_{l,m}) =  \cup_{i=1}^l\cap_{j=1}^m \Reach(U_{i,j}),$$ 
a \emph{\unioninter\  safety} objective $\UISafe(U_{1,1},\ldots,U_{l,m})$, or a \emph{\unioninter\  B\"uchi} objective $\UIBuchi(U_{1,1},\ldots,U_{l,m})$.
\end{itemize}

Games $(G,\Obj)$ with $\omega$-regular objectives $\Obj$ are determined by Martin's theorem~\cite{Martin75}. We recall the complexity class of solving those games, as well as the kind (memoryless, finite-memory) of winning strategies for both players. See Theorem~\ref{thm:gameresults} and Table~\ref{table:classicalobjectives} below. For each type of objective, the complexity of the algorithms is expressed in terms of the sizes $|V|$ and $|E|$ of the game structure $G$, the number $d$ of colors (for $\Par$ and $\Muller$), the number $k$ of pairs (for $\Rabin$ and $\Streett$), the size $|\mathcal{F}|$ of the family $\mathcal{F}$ (for $\EMuller$ and $\Muller$), the size $|\phi|$ of the formula $\phi$ (for $\BoolBuchi$), the number $m$ of intersections of objectives (for $\GenReach$ and $\GenBuchi$), and the number $m$ (resp. $l$) of intersections (resp. unions) in \unioninter\ objectives (for $\UIReach, \UISafe$, and $\UIBuchi$). 

\begin{theorem}\label{thm:gameresults} For games $(G,\Obj)$ with $\omega$-regular objectives, we have:
\begin{itemize}
\item Solving reachability or safety games is $\mathsf P$-complete (with an algorithm in $O(|V|+|E|)$ time) and both players have memoryless winning strategies {\rm\cite{Beeri80,2001automata,Immerman81}}.
\item Solving B\"uchi or co-B\"uchi games is $\mathsf P$-complete (with an algorithm in $O(|V|^2)$ time)  and both players have memoryless winning strategies  {\rm\cite{ChatterjeeH14,EmersonJ91,Immerman81}}. 
\item Solving explicit Muller games with a family $\mathcal F$ is $\mathsf{P}$-complete (with an algorithm in $O(|\mathcal{F}| \cdot (|\mathcal{F}| + |V| \cdot |E|)^2)$ time) and exponential memory strategies are necessary and sufficient for both players {\rm\cite{DziembowskiJW97,Horn08}}.
\item Solving Rabin (resp. Streett) games with $k$ pairs is $\mathsf{NP}$-complete (resp. co-$\mathsf{NP}$-complete) {\rm\cite{EmersonJ88}} (with an algorithm in $O(|V|^{k+1} \cdot k!)$ time~{\rm\cite{PitermanP06}}). In Rabin games (resp. Streett games) memoryless strategies are sufficient for $\playerOne$ (resp. for $\playerTwo$) {\rm\cite{Emerson85}} and exponential memory strategies are necessary and sufficient for $\playerTwo$ (resp. $\playerOne$) {\rm \cite{DziembowskiJW97}} 
\item Solving parity games with $d$ colors is in $\mathsf{UP}\cap\mathsf{co}$-$\mathsf{UP}$ (with an algorithm in $O(|V|^{\lceil \log(d) \rceil+6})$ time~{\rm\cite{Calude}}) and both players have memoryless winning strategies {\rm\cite{Jurdzinski98}}.
\item Solving Muller games is $\mathsf{PSPACE}$-complete (with an algorithm in $O(|V|^2 \cdot |E| \cdot |V|!)$ time~{\rm \cite{McNaughton93}}) and exponential memory strategy are necessary and sufficient for both players {\rm\rm{\cite{DziembowskiJW97,HunterD05}}}.
\item Solving Boolean B\"uchi games is $\mathsf{PSPACE}$-complete (with an algorithm in $O(|\phi| \cdot 2^{O(|V|^2)})$ time and exponential memory strategies are necessary and sufficient for both players {\rm \cite{AlurTM03}}.\footnote{The algorithm complexity and the memory requirements do not appear explicitly in {\rm\cite{AlurTM03}} but can be deduced straightforwardly thanks to the proposed algorithm.}
\item Solving generalized reachability games with $m$ target sets is $\mathsf{PSPACE}$-complete (with an algorithm in $O(2^m\cdot (|V| + |E|))$ time) and exponential memory strategies are necessary and sufficient for both players~{\rm\cite{FijalkowH13}}.
\item Solving generalized B\"uchi games with $m$ target sets is $\mathsf P$-complete (with an algorithm in $O(m \cdot |V|^2)$ time) and linear memory (resp. memoryless) strategies are necessary and sufficient for $\playerOne$ (resp. $\playerTwo$)  {\rm\cite{ChatterjeeDHL16}}.
\item Solving \unioninter\ reachability and \unioninter\ safety objectives is $\mathsf{PSPACE}$-complete (with an algorithm in $O(2^K\cdot (|V| + |E|))$ time) and exponential memory strategies are necessary and sufficient for both players, where $K$ denotes the number of distinct target sets.
\item Solving \unioninter\ B\"uchi games with an objective $\cup_{i=1}^l\cap_{j=1}^m \Buchi(U_{i,j})$ is $\mathsf{coNP}$-complete (with an algorithm in $O(m^l \cdot |V|^2)$ time), and exponential memory (resp. memoryless) strategies are necessary and sufficient for $\playerOne$ (resp. $\playerTwo$) {\rm\cite{BloemCGHJ10}}.
\end{itemize}
\end{theorem}
\begin{proof}
All the statements follow from the literature except for the case of \unioninter\ reachability and \unioninter\ safety games for which we provide a proof. We only consider the reachability case, since the proof is similar for the safety case. First, as solving \unioninter\ reachability games is harder than solving generalized reachability games (when there is no union), we immediately obtain the lower bounds for the complexity and the memory requirements. Indeed, solving generalized reachability games is $\mathsf{PSPACE}$-complete, and exponential memory strategies are necessary for both players~\cite{FijalkowH13}. 

Let us now prove the upper bounds by following the same approach as proposed in~\cite{FijalkowH13} to solve generalized reachability games. Let $(G,\Obj)$ be a \unioninter\ reachability game where $\Obj = \cup_{i=1}^l \cap_{j=1}^m \Reach(U_{i,j})$. We define the function $f' \colon \{ U_{i,j} \mid i \in \{1,\ldots,l\}, j \in \{1,\ldots,m\}\} \rightarrow \{1,\ldots,K\}$ that enumerates all the distinct sets $U_{i,j}$. From $f'$, we construct the function $f \colon \{1,\ldots,l\} \times \{1,\ldots,m\} \rightarrow \{1,\ldots,N\}$ such that $f(i,j) = k$ if $f'(U_{i,j}) = k$. If $f(i,j) = k$, we abusively write $U_{i,j} = U_k$. 

We construct from $G = (V_1,V_2,E)$ a new game structure $G'$ $=$ $(V'_1,V'_2,E')$ in a way to remember which sets $U_k$ have been visited so far, for $k \in \{1,\ldots,K\}$. Formally, $V_i' = V_i \times \{0,1\}^{K}$ for $i \in \{1,2\}$, and $((v,b_1,\ldots,b_K),(v',b'_1,\ldots,b'_K)) \in E'$ if and only if $(v,v') \in E$ and for all $k$, $b'_{k} = 1$ if $b_{k} =1$ or $v' \in U_k$, and $0$ otherwise. With the initial vertex $v_0$ in $G$, we associate the initial vertex $(v_0,b^0_1,\ldots,b^0_K)$ in $G'$ where $b^0_{k} = 1$ if $v_0 \in U_{k}$ and $0$ otherwise. We then have that $\playerOne$ is winning in the original \unioninter\ reachability game from $v_0$ if and only if $\playerOne$ is winning in $G'$ from $(v_0,b^0_1,\ldots,b^0_K)$ for the objective $\Reach(U)$  where $U = \{(v,b_1,\ldots,b_K) \mid \exists i \in \{1,\ldots,l\}, \forall j \in \{1,\ldots,m\}, b_{f(i,j)} = 1\}$. 

Note that solving this reachability game $(G',\Reach(U))$ can be done in time linear in the size of the game with memoryless winning strategies for both players~by~\cite{2001automata}. Coming back to the initial \unioninter\ reachability game, this leads to an algorithm working in $O(2^K \cdot (|V|+|E|))$ time, and to exponential memory winning strategies for both players.

Now, as done for generalized reachability games~\cite{FijalkowH13}, one can notice that if $\playerOne$ is winning for  $\Reach(U)$, then he has a strategy to do so within $K \cdot |V|$ steps. Moreover, given a path of this size, one can check in polynomial time if there exists some $i$ such that the path visits all $U_{f(i,j)}$ for $j \in \{1,\ldots,m\}$. Thus, we can use an alternating Turing machine that simulates the game for up to $K \cdot |V|$ steps and checks whether $\playerOne$ is winning. As the alternating Turing machine works in polynomial time and $\mathsf{APTIME} = \mathsf{PSPACE}$, this yields the $\mathsf{PSPACE}$ algorithm.
\qed
\end{proof}

\begin{table}
\normalsize
\begin{center}
\begin{tabular}{|c|c|c|c|}
\hline
Objectives & ~Complexity class~  & ~~$\playerOne$ memory~~ &~~ $\playerTwo$ memory ~~\\
\hline
\hline
Reachability, safety & \multirow{4}{*}{$\mathsf{P}$-complete}   &  \multicolumn{2}{c|}{\multirow{2}{*}{memoryless}}  \\
\cline{0-0}
B\"uchi, co-B\"uchi  &   & \multicolumn{2}{c|}{} \\
\cline{0-0} \cline{3-4}
Explicit Muller & & \multicolumn{2}{c|}{exponential} \\
\cline{0-0} \cline{3-4}
Generalized B\"uchi &  & linear & memoryless \\
\hline
Generalized reachability & \multirow{2}{*}{$\mathsf{PSPACE}$-complete} &  \multicolumn{2}{c|}{\multirow{2}{*}{exponential}} \\
\cline{0-0}
\unioninter\ reachability, \unioninter\ safe &  & \multicolumn{2}{c|}{}  \\
\hline
Parity & $\mathsf{NP} \cap \mathsf{coNP}$ & \multicolumn{2}{c|}{memoryless} \\
\hline
Rabin & $\mathsf{NP}$-complete & memoryless & exponential  \\
\hline
Streett & \multirow{2}{*}{$\mathsf{coNP}$-complete} & \multirow{2}{*}{exponential} & \multirow{2}{*}{memoryless} \\
\cline{0-0}
\unioninter\ B\"uchi&  &  &  \\
\hline
Muller& \multirow{2}{*}{$\mathsf{PSPACE}$-complete} & \multicolumn{2}{c|}{\multirow{2}{*}{exponential}}  \\
\cline{0-0}
~Boolean B\"uchi & & \multicolumn{2}{c|}{} \\
\hline
\end{tabular}
\end{center}
\caption{Overview of results on games with a single $\omega$-regular objective. The last two columns indicate the tight memory requirements of the winning strategies.}
\label{table:classicalobjectives}
\end{table}

In the sequel, we need some classical properties on $\omega$-regular objectives that we summarize in the following proposition.

\begin{proposition}\label{rem:objectives} 
\begin{enumerate}
\item \label{item:safety} A safety (resp. co-B\"uchi, Streett) objective is the complement of a reachability (resp. B\"uchi, Rabin) objective. 
\item \label{item:parityRabin} A parity objective is both a Rabin and a Streett objective.
\item \label{item:Streett} Rabin and Streett objectives with one pair are parity objectives with $3$ colors. Thus, a Rabin (resp. Streett) objective is the union (resp. intersection) of parity objectives with $3$ colors.
\item \label{item:ExplMuller} The intersection of $m$ (resp. union of $l$) explicit Muller objectives $\EMuller(\mathcal{F}_i)$ is an explicit Muller objective $\EMuller(\mathcal{F})$ where $|\mathcal{F}| \leq \min_{i \in \{1,\ldots,m\}}\{|\mathcal{F}_i|\}$ (resp. $|\mathcal{F}| \leq \sum_{i=1}^l |\mathcal{F}_i|$).
\item \label{item:parity} A parity objective with $d$ colors (resp. Streett objective with $k$ pairs, Rabin objective with $k$ pairs, Muller objective with $d$ colors and a family $\mathcal{F}$) is a Boolean B\"uchi objective of size at most $\frac{d^2}{2}$ (resp. $2 \cdot k$, $d \cdot |\mathcal{F}|$) and with $d$ (resp. $2 \cdot k$, $d$) variables. 
\end{enumerate}
\end{proposition}

\begin{proof}
First, Item $\ref{item:safety}$ 
immediately follows from the definitions and Items~\ref{item:parityRabin} and~\ref{item:Streett} 
are detailed in \cite{ChatterjeeHP07}.

Let us consider Item~\ref{item:ExplMuller}. For the intersection we have $\cap_{i=1}^m \EMuller(\mathcal{F}_i) = \EMuller(\mathcal{F})$ where $\mathcal{F} = \cap_{i=1}^m \mathcal{F}_i$, and thus $|\mathcal{F}| \leq \min_{i \in \{1,\ldots,m\}}\{|\mathcal{F}_i|\}$. For the union we have $\cup_{i=1}^l \EMuller(\mathcal{F}_i) = \EMuller(\mathcal{F})$ with $\mathcal{F} = \cup_{i=1}^l \mathcal{F}_i$ with $|\mathcal{F}| \leq \sum_{i=1}^l |\mathcal{F}_i|$.

Let us prove the last item by beginning with Muller objectives. It suffices to note that a play belongs to $\Muller(p,\mathcal{F})$ if and only if there exists an element $F$ of $\mathcal{F}$ such that all colors of $F$ are seen infinitely often along the play while no other color is seen infinitely often. This is obviously a Boolean B\"uchi objective $\BoolBuchi(\phi,U_1,\ldots,U_m)$, where each $U_i$ corresponds to a color, that is, $U_i$ is the set of vertices labeled by this color. Note that, in this case, the size of the related formula $\phi$ is at most $d \cdot |\mathcal{F}|$. The arguments are similar for parity, Streett and Rabin objectives (for instance, a play belongs to $\Par(p)$ if and only there exists an even color seen infinitely often along the play and no lower color seen infinitely often).     
\qed
\end{proof}
\section{Fixed parameter complexity of ordered $\omega$-regular games} \label{sec:FPT} 

In this section, we study the fixed parameter tractability of the \tproblem.

\paragraph{\bf Parameterized complexity.} 

A \emph{parameterized language} $L$ is a subset of $\Sigma^* \times \N$, where $\Sigma$ is a finite alphabet, the second component being the parameter of the language. It is called \textit{fixed parameter tractable} (FPT) if there is an algorithm that determines whether $(x,t) \in L$ in time $f(t) \cdot |x|^c$ time, where $c$ is a constant independent of the parameter $t$ and $f$ is a computable function depending on $t$ only. We also say that $L$ belongs to (the class) $\FPT$. Intuitively, a language is FPT if there is an algorithm running in polynomial time w.r.t the input size times some computable function on the parameter. 
In this framework, we do not rely on classical polynomial reductions but rather use so called $\FPT$-reductions. An \emph{$\FPT$-reduction} between two parameterized languages $L \subseteq \Sigma^* \times \N$ and $L' \subseteq \Sigma'^* \times \N$ is a function $R : L \to L'$ such that 
\begin{itemize}
\item $(x,t) \in L$ if and only if $(x',t') = R(x,t) \in L'$,
\item $R$ is computable by an algorithm that takes $f(t) \cdot |x|^c$ time where $c$ is a constant, and
\item $t' \leq g(t)$ for some computable function $g$.
\end{itemize}
Moreover, if $L'$ is in $\FPT$, then $L$ is also in $\FPT$. 
We refer the interested reader to \cite{DowneyF99} for more details on parameterized complexity.

Our main result states that the threshold problem is in $\FPT$ for all the ordered games of this article. Parameterized complexities are given in Table~\ref{table:FPT}.

\begin{theorem}\label{thm:lexiFPT}
The \tproblem\ is in $\FPT$ for ordered reachability, safety, B\"uchi, co-B\"uchi, explicit Muller, Rabin, Streett, parity, and Muller games.
\end{theorem}

\medskip

\begin{table}
\scriptsize
\begin{center}
\begin{tabular}{|c||c|c|c|}
\hline
~Objectives~ & ~Parameters~ & ~Threshold problem~  \\
\hline
\rule{0cm}{0.25cm} ~Reachability, Safety~ & $n$ & $O(\s \cdot n + 2^n \cdot (|V| + |E|))$ \\
\hline
\rule{0cm}{0.25cm} ~B\"uchi~ & $n$ & $O(\s \cdot n + \sprim \cdot |V|^2)$ \\
\hline
\rule{0cm}{0.25cm} ~co-B\"uchi~ & $n$ & $O(\s \cdot n + \s \cdot |V|^2)$ \\
\hline
\rule{0cm}{0.25cm} ~Explicit Muller~ & $n$ & $O(\s \cdot n + (\s \cdot \max_i |{\cal F}_i|)^3 \cdot |V|^2 \cdot  |E|^2)$ \\
\hline
\rule{0cm}{0.3cm} ~Rabin, Streett~  & $n$, $k_1, \ldots, k_n$ & $O((2^{M_1} \cdot N_1 + M_1^{M_1}) \cdot |V|^5)$ \\
\hline
\rule{0cm}{0.3cm} Parity  & ~$n$, $d_1, \ldots, d_n$~ & $O((2^{M_2} \cdot N_2 + M_2^{M_2}) \cdot |V|^5)$  \\
\hline
\rule{0cm}{0.3cm} Muller  & ~$n$, $d_1, \ldots, d_n$~ & $O((2^{M_3} \cdot N_3 + M_3^{M_3}) \cdot |V|^5)$ \\
\hline
\end{tabular}\end{center}
\caption{Fixed parameter tractability of ordered games $(G,\Obj_1,\ldots, \Obj_n,\precsim)$: for $i \in \{1,\ldots,n\}$, $k_i$/$d_i$ denotes the number of pairs/colors of each Rabin/Streett/Muller objective $\Obj_i$. Sizes $\s$ and $\sprim$ are resp. upper bounded by $2^n$ and $2^{2^n}$.
For $j \in \{1,2,3\}$, $M_j = 2^{m_j}$ where $m_1 = \sum_{i=1}^n 2 \cdot k_i$, $m_2 = m_3 = \sum_{i=1}^n d_i$, and $N_1 = \s \cdot \sum_{i =1}^n 2 \cdot k_i$, $N_2 = \s \cdot \sum_{i =1}^n \frac{d_i^2}{2}$, $N_3 = \s \cdot \sum_{i =1}^n 2^{d_i} \cdot d_i$.}
\label{table:FPT}
\end{table}

The proof of this theorem needs to show that solving the \tproblem\ for an ordered game $(G,\Obj_1,\ldots, \Obj_n,\precsim)$ is equivalent to solving a game $(G,\Obj)$ with a single objective $\Obj$ equal to the union of intersections of objectives taken in $\{\Obj_1,\ldots,\Obj_n\}$. It also needs to show that solving Boolean B\"uchi games is in $\FPT$.

\paragraph{\bf Monotonic preorders embedded in the subset preorder.}
 
We here present a \textit{key tool} of this paper: solving the \tproblem\ for an ordered game $(G,\Obj_1,\ldots, \Obj_n,\precsim)$ is equivalent to solving a game $(G,\Obj)$ with a single objective $\Obj$ equal to the union of intersections of objectives taken in $\{\Obj_1,\ldots,\Obj_n\}$. The arguments are the following ones. (1) We consider the set $\{0,1\}^n$ of payoffs ordered with $\precsim$ as well as ordered with the subset preorder $\subseteq$ (see the example of Figure~\ref{fig:plongement} where $\precsim$ is the lexicographic preorder). To any payoff $\nu \in \{0,1\}^n$, we associate the set $\delta_{\nu} = \{ i \in \{1, \ldots, n\} \mid \nu_i = 1\}$ containing all indices $i$ such that objective $\Obj_i$ is satisfied. 
(2) Consider the set of payoffs $\nu \succsim \mu$ embedded in the set $\{0,1\}^n$ ordered with $\subseteq$. By monotonicity of $\precsim$, we obtain an upper-closed set $\Set$ that can be represented by the antichain of its \emph{minimal elements} (with respect to $\subseteq$), that we denote by $\M(\mu)$. (3) $\playerOne$ can ensure a payoff $\succsim \mu$ if and only if he has a strategy such that any consistent outcome $\pi$ has a payoff $\nu^* \supseteq \nu$ for some $\nu \in \M(\mu)$, equivalently such that $\pi$ satisfies (at least) the conjunction of the objectives $\Obj_i$ such that $\nu_i = 1$. (4) The objective $\Obj$ of $\playerOne$ is thus a disjunction (over $\nu \in \M(\mu)$) of conjunctions (over $i \in \delta_{\nu}$) of objectives $\Obj_i$. This statement is formulated in the next theorem (see again Figure~\ref{fig:plongement}).

\begin{theorem}\label{thm:reduction}
Let $(G,\Obj_1,\ldots, \Obj_n,\precsim)$ be an ordered game, $\mu \in \{0,1\}^n$ be some threshold, and $v_0$ be an initial vertex. Then, $\playerOne$ can ensure a payoff $\succsim \mu$ from $v_0$ in $(G,\Obj_1,\ldots, \Obj_n,\precsim)$ if and only if $\playerOne$ has a winning strategy from $v_0$ in the game $(G,\Obj)$ with the objective
$\Obj = \cup_{\nu \in \M(\mu)} \cap_{i \in \delta_{\nu}} \Obj_i$.
\qed
\end{theorem}

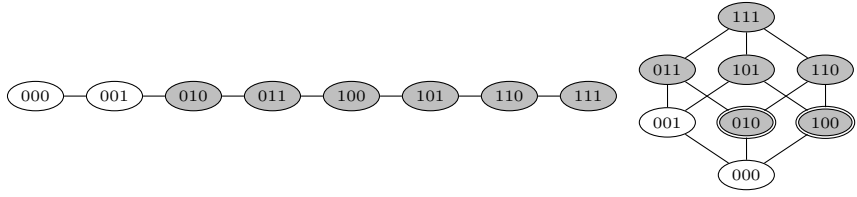
\begin{figure}[t]
   \centering
    \begin{tikzpicture}[scale=3.5]
    \everymath{\scriptstyle}
    
    \draw (0,0) node [ellipse, inner sep =2pt, draw] (A) {$000$};
    \draw (0.3,0) node [ellipse, inner sep =2pt, draw] (B) {$001$};
    \draw (0.6,0) node [ellipse, fill=gray!50,inner sep =2pt, draw] (C) {$010$};
    \draw (0.9,0) node [ellipse, fill=gray!50,inner sep =2pt, draw] (D) {$011$};
    \draw (1.2,0) node [ellipse, fill=gray!50,inner sep =2pt, draw] (E) {$100$};
    \draw (1.5,0) node [ellipse, fill=gray!50,inner sep =2pt, draw] (F) {$101$};
    \draw (1.8,0) node [ellipse, fill=gray!50,inner sep =2pt, draw] (G) {$110$};
    \draw (2.1,0) node [ellipse, fill=gray!50,inner sep =2pt, draw] (H) {$111$};
    \draw (A) -- (B) -- (C) -- (D) -- (E) -- (F) -- (G) -- (H);

    \draw (2.7,-0.3) node [ellipse, inner sep =2pt, draw] (AA) {$000$};
    \draw (2.7,-0.1) node [ellipse, double, fill=gray!50, inner sep =2pt, draw] (CC) {$010$};
    \draw (2.4,-0.1) node [ellipse, inner sep =2pt, draw] (BB) {$001$};
    \draw (3,-0.1) node [ellipse, double, fill=gray!50, inner sep =2pt, draw] (DD) {$100$};
    \draw (AA) -- (CC);
    \draw (AA) -- (BB);
    \draw (AA) -- (DD);
    \draw (2.4,0.1) node [ellipse, fill=gray!50, inner sep =2pt, draw] (EE) {$011$};
    \draw (2.7,0.1) node [ellipse, fill=gray!50, inner sep =2pt, draw] (FF) {$101$};
    \draw (3,0.1) node [ellipse, fill=gray!50, inner sep =2pt, draw] (GG) {$110$};
    \draw (CC) -- (EE);
    \draw (CC) -- (GG);
    \draw (BB) -- (EE);
    \draw (BB) -- (FF);
    \draw (DD) -- (FF);
    \draw (DD) -- (GG);
    
    \draw (2.7,0.3) node [ellipse, fill=gray!50, inner sep =2pt, draw] (HH) {$111$};
    \draw (FF) -- (HH);
    \draw (GG) -- (HH);
    \draw (EE) -- (HH);
      
    \end{tikzpicture} 
    \caption{Gray nodes represent the set of payoffs $\nu \succsim \mu = 010$ for the lexicographic preorder and its embedding for the subset preorder. The elements of $\M(\mu) = \{010, 100 \}$ are doubly circled nodes.}
   \label{fig:plongement}
\end{figure}

Note that we obtain the following corollary as a direct consequence of Theorem~\ref{thm:reduction} and Martin's theorem~\cite{Martin75}.

\begin{corollary}\label{cor:deter}
Let $(G,\Obj_1,\ldots,\Obj_n)$ be an ordered game. 
If $\Obj_1, \ldots, \Obj_n$ are Borel sets, then $\playerOne$ has a strategy to ensure a payoff $\succsim \mu$ from $v_0$ if and only if it is not the case that $\playerTwo$ has a strategy to avoid a payoff $\succsim \mu$ from~$v_0$.
\qed\end{corollary}

\medskip\noindent{\bf Parameterized complexity of Boolean B\"uchi games.~} 
In order to show that solving the \tproblem\ for ordered games is in $\FPT$, we need to recall some known results of parameterized complexity for games with a single objective and to prove that solving Boolean B\"uchi games belongs to $\FPT$. 

It is proved in~\cite{FijalkowH13} that generalized reachability games belong to $\FPT$. Parity, Rabin, Streett, and Muller games are shown to be $\FPT$-interreducible in~\cite{FPT}. Very recently, Calude and al. provided a quasipolynomial time algorithm for parity games and showed that parity games are in $\FPT$~\cite{Calude}. It follows that Rabin, Streett, and Muller games also belong to $\FPT$. All these results are summarized in the next theorem with the related complexities.

\begin{theorem}\label{thm:FPT}
Solving generalized reachability, parity, Rabin, Streett, and Muller games is in $\FPT$. Generalized reachability (resp. parity, Muller) games are solvable with an algorithm running in $O(2^m \cdot (|V| + |E|))$ (resp. $O(|V|^5) + g(d)$, $O((d^d \cdot |V|)^5)$) time, where parameter $m$ is the number of reachability objectives, parameter $d$ is the number of colors, and $g$ is some computable function.
\end{theorem}

\begin{proposition}\label{prop:FPT}
Solving Boolean B\"uchi games $(G,\Obj)$ is in $\FPT$, with an algorithm in $O(2^{M} \cdot |\phi| + (M^M \cdot |V|)^5)$ time with $M = 2^m$ such that $m$ is the number of variables of $\phi$ in the Boolean B\"uchi objective $\Obj$.
\end{proposition}

\begin{proof}
Let us show the existence of an $\FPT$-reduction from Boolean B\"uchi games to Muller games. For this purpose, consider a Boolean B\"uchi game $(G,\Obj)$ with the objective $\Obj = \BoolBuchi(\phi,U_1,\ldots,U_m)$, where $\phi$ is a Boolean formula over variables $x_1,\ldots,x_m$, and $m$ is seen as a parameter. We build an adequate Muller game $(G,\Muller(p,\mathcal{F}))$ on the same game structure and parameterized by the number of colors. The coloring function $p$ and the family $\mathcal F$ are constructed as follows.

To any vertex $v \in V$, we associate a color $p(v) = \mu$ which is a subset of $\{1,\ldots,m\}$ in the following way: $i \in \mu$ if and only if  $v \in U_i$.\footnote{Our definition of color requires $\mu$ to be an integer. It suffices to associate with $v$ a vector $\mu^v \in \{0,1\}^m$ such that $\mu^v_i = 1$ if $v \in U_i$ and $0$ otherwise, and to define the coloring function $p \colon V \rightarrow  \{0,\ldots, 2^m-1\}$ that associates with each vertex $v$ the color $p(v)$ such that its binary encoding is equal to $\mu^v$.} Intuitively, we keep track for all $i$, whether a vertex belongs to $U_i$ or not. The total number $M$ of colors is thus equal to $2^m$. One can notice that $(*)$ a play $\pi$ visits a vertex $v \in U_i$ if and only if $\pi$ visits a color $\mu$ that contains $i$.  

To any subset $F$ of $p(V)$, we associate the truth assignment $\chi(F) \in \{0,1\}^m$ of variables $x_1, \ldots, x_m$ such that for all $i$, $\chi(F)_i = 1$ if there exists $\mu \in F$ such that $i \in \mu$, and $0$ otherwise. The idea (by $(*)$) is that the set $F$ of colors visited infinitely often by a play $\pi$ corresponds to the set $\Occinf(\pi)$ of vertices visited infinitely often such that  $\chi(F)_i = 1$ if and only if $\Occinf(\pi) \cap U_i \ne \emptyset$. We then define $\mathcal{F} = \{ F \subseteq p(V) \mid \chi(F) \models \phi \}$, that is, $\mathcal{F}$ corresponds to the set of all truth assignments satisfying $\phi$.

In this way we have the desired $\FPT$-reduction: first, parameter $M = 2^m$ only depends on parameter $m$. Second, we have that $\playerOne$ is winning in $(G,\BoolBuchi(\phi,U_1,\ldots,U_m))$ from an initial vertex $v_0$ if and only if he is winning in $(G,\Muller(p,\mathcal{F}))$ from $v_0$. Indeed, a play $\pi$ satisfies $(\phi,U_1,\ldots,U_m)$ if and only if the truth assignment (~$x_i = 1$ if and only $\Occinf(\pi) \cap U_i \neq \emptyset$, and $x_i = 0$ otherwise~) satisfies $\phi$. This is equivalent to have that $F = p(\Occinf(\pi))$ belongs to $\mathcal F$ (by definition of $\chi(F)$), that is, $\pi$ belongs to $\Muller(p,{\mathcal F})$. Third, the construction of the Muller game is in $O(2^{2^m} \cdot |\phi|)$ time since it requires $O(|V| + |E|)$ time for the game structure,  $O(m \cdot |V|)$ time for the coloring function $p$, and $O(2^{2^m} \cdot |\phi|)$ time for the family~$\mathcal F$. 

From this $\FPT$-reduction and as solving Muller games is in $O((d^d \cdot |V|)^5)$) time where $d$ is the number of colors~\cite{Calude}, we have an algorithm solving the Boolean B\"uchi game in $O(2^{M} \cdot |\phi| + (M^M \cdot |V|)^5)$ time, where $M = 2^m$.
\qed\end{proof}

\medskip\noindent{\bf Proof of Theorem~\ref{thm:lexiFPT}.~} By Theorem~\ref{thm:reduction}, we know that solving the \tproblem\ for an ordered game $(G,\Obj_1,\ldots, \Obj_n,\precsim)$ is equivalent to solving a game $(G,\Obj)$ with a single objective $\Obj = \cup_{\nu \in \M(\mu)} \cap_{i \in \delta_{\nu}} \Obj_i$. Thanks to this equivalence, we provide a proof of Theorem~\ref{thm:lexiFPT} with the parameterized complexities given in Table~\ref{table:FPT}. This proof uses two sizes depending on the number $n$ of objectives:
\begin{itemize}
\item the \emph{size $\s$} of $\M(\mu)$. It is upper bounded by $2^n$ (an antichain of maximum size in the subset preorder over $\{0,1\}^n$ is of exponential size $\binom{n}{\lfloor \frac{n}{2} \rfloor}$).
\item the \emph{size $\sprim$} defined as follows. In case of B\"uchi objectives $\Obj_i$, we need to rewrite the objective $\cup_{\nu \in \M(\mu)} \cap_{i \in \delta_{\nu}} \Obj_i$ in conjunctive normal form $\cap_{k} \cup_{l} \Obj'_{k,l}$ with $\Obj'_{k,l} \in \{\Obj_1,\ldots,\Obj_n\}$. We denote by $\sprim$ the size of this disjunction. It is bounded by $2^{2^n}$. 
\end{itemize}

In Section~\ref{sec:homo} we will show that, for several objectives, we can go beyond the fixed parameter tractability of Theorem~\ref{thm:lexiFPT} by providing polynomial time algorithms when the sizes $\s$ and $\sprim$ are polynomial in $n$.

\begin{proof}[of Theorem~\ref{thm:lexiFPT}] 
By Theorem~\ref{thm:reduction}, solving the \tproblem\ for an ordered game $(G,\Obj_1,\ldots, \Obj_n,\precsim)$ is equivalent to solving a classical game $(G,\Obj)$ with $\Obj = \cup_{\nu \in \M(\mu)} \cap_{i \in \delta_{\nu}} \Obj_i$. We have $|\M(\mu)| = \s$ and $|\delta_{\nu}| \leq n$ $\forall \nu \in \M(\mu)$. Recall that $\s \leq 2^n$ and $\sprim \leq 2^{2^n}$.

We first show that the \tproblem\ for ordered reachability, safety, B\"uchi, co-B\"uchi, and explicit Muller games is in $\FPT$ with parameter $n$. The reduction provided in Theorem~\ref{thm:reduction} is an $\FPT$-reduction as the number of disjunctions/conjunctions in $\Obj$ only depends on $n$. Moreover the construction of the game $(G,\Obj)$ is in $O(|V| + |E| + \s \cdot n)$ time. In the following items we describe a second $\FPT$-reduction to add to the first one. The sum of the complexities of both $\FPT$-reductions leads to the complexities given in Table~\ref{table:FPT}, rows 2-5.

\begin{itemize}
\item If each $\Obj_i$ is a reachability (resp. safety) objective, then $(G,\Obj)$ is a \unioninter\ reachability (resp. safety) game that can be reduced to a reachability (resp. safety) game over a game structure of size $2^n \cdot |V|$ by Theorem~\ref{thm:gameresults}. The latter is solved in $O(2^n \cdot (|V|+ |E|))$ time. 

\item If $\Obj$ is a union of intersections of B\"uchi objectives, then it can be rewritten as the intersection of unions of B\"uchi objectives which is a generalized B\"uchi objective with at most $\sprim$ target sets. The latter game is solved in $O(\sprim \cdot |V|^2)$ time by Theorem~\ref{thm:gameresults}. The union of intersections of co-B\"uchi objectives is the complementary of a generalized B\"uchi objective with at most $\s$ target sets, leading to an algorithm in $O(\s \cdot |V|^2)$ time.

\item If each $\Obj_i$ is an explicit Muller objective $\EMuller(\mathcal{F}_{i})$ then $\Obj$ is also an explicit Muller objective. Indeed the intersection (resp. union) of explicit Muller objectives is an explicit Muller objective such that $\cap_i \EMuller(\mathcal{F}_i) = \EMuller(\mathcal{F})$ with $\mathcal{F} = \cap_i \mathcal{F}_i$ (resp. $\cup_i \EMuller(\mathcal{F}_i) = \EMuller(\mathcal{F})$ with $\mathcal{F} = \cup_i \mathcal{F}_i$). Therefore $\Obj$ can be here rewritten as $\EMuller(\mathcal{F})$ for some set $\mathcal{F}$ such that $|\mathcal{F}| \leq \sum_{\nu \in \M(\mu)}\min_{j \in \delta_{\nu}} |\mathcal{F}_j |$. The latter game is solved in $O(|\mathcal{F}| \cdot (|V| \cdot |E| + |\mathcal{F}|)^2) = O((\s \cdot \max_i |{\cal F}_i|)^3 \cdot |V|^2 \cdot  |E|^2)$ time by Theorem~\ref{thm:gameresults}. 
\end{itemize}

We now show that the \tproblem\ for ordered parity, Rabin, Streett, and Muller games is in $\FPT$ thanks to Proposition~\ref{prop:FPT}. 

\begin{itemize}
\item  Let us show that the \tproblem\ for ordered parity games is in $\FPT$ with parameters $n, d_1, \ldots, d_n$. If each $\Obj_i$ is a parity objective with $d_i$ colors, then each $\Obj_i$ is a Boolean B\"uchi objective of size at most $\frac{d_i^2}{2}$ and using $d_i$ variables. Indeed, as a play is winning for $\Obj_i$ if and only there exists an even priority seen infinitely often along the play and no lower priority seen infinitely often. Therefore, $\Obj$ is a Boolean B\"uchi objective $\Obj'$ of size $|\phi| \leq \s \cdot \sum_{i=1}^n \frac{d_i^2}{2}$, and with $m = \sum_{i=1}^n d_i$ variables as $\cup_{\nu \in \M(\mu)} \{\Obj_i \mid i \in \delta_{\nu} \} \subseteq \{\Obj_1,\ldots, \Obj_n\}$. We thus have an $\FPT$-reduction to the game $(G,\Obj')$ depending on the parameters $n, d_1, \ldots, d_n$ and with an algorithm in $O(|V| + |E| + |\phi|)$ time.
By Proposition~\ref{prop:FPT}, solving the game $(G,\Obj')$ is in $\FPT$ with an algorithm in $O(2^{M} \cdot |\phi| + (M^M \cdot |V|)^5)$ time with $M = 2^m$. Thus the \tproblem\ is in $\FPT$ with parameters $n, d_1, \ldots, d_n$, with an overall algorithm in $O((2^{M} \cdot N + M^{M}) \cdot |V|^5)$ time where $N = 2^n \cdot \sum_{i=1}^n \frac{d_i^2}{2}$.

\item The arguments are similar for ordered Rabin, Streett, and Muller games. The only differences are the bound on size $|\phi|$, the number $m$ of variables and the parameters (see Table~\ref{table:FPT}). For Rabin and Streett games, the parameters are $n, k_1,\ldots, k_n$ and we have $|\phi| \leq \s \cdot \sum_{i =1}^n 2 \cdot k_i = N$ and $m = \sum_{i=1}^n 2 \cdot k_i$. For Muller games, the parameters are $n, d_1, \ldots, d_n$ and we have $|\phi| \leq \s \cdot \sum_{i =1}^n |\mathcal{F}_i| \cdot d_i \leq \s \cdot \sum_{i =1}^n 2^{d_i} \cdot d_i  = N$, where $\forall i$ $m = \sum_{i=1}^n d_i$.

\end{itemize}
\qed\end{proof}
\section{Ordered games with a compact embedding}\label{sec:homo}
 
In the previous section, we have shown that solving the threshold problem for ordered $\omega$-regular games is in $\FPT$. This result depends on sizes $\s$ and $\sprim$ which vary with the number $n$ of objectives.  
In this section, we study ordered games such that these sizes are polynomial in $n$. 

\paragraph{\bf Preorders with a compact embedding in the subset preorder.~}
An ordered game $(G,\Obj_1,\ldots, \Obj_n,\precsim)$ has a \emph{compact embedding} (in the subset preorder) if both sizes $\s$ and $\sprim$ are polynomial in $n$. While the threshold problem is in $\FPT$ for ordered B\"uchi, co-B\"uchi, and explicit Muller games, it becomes polynomial as soon as their preorder has a compact embedding. This is a direct consequence of Table~\ref{table:FPT}, rows 2-4.

\begin{theorem} \label{thm:poly}
The threshold problem is solved in polynomial time for ordered B\"uchi, co-B\"uchi, and explicit Muller games with a compact embedding.
\end{theorem}

One can easily prove that ordered games using the subset or the maximize preorder have a compact embedding. We will later prove that this also holds for the lexicographic preorder. Nevertheless it is not the case for the counting preorder. Indeed solving the threshold problem for counting B\"uchi games is {\sf co-NP}-complete~\cite{BouyerBMU12}. 

Recall that solving the threshold problem for ordered B\"uchi games reduces to solving some \unioninter\ B\"uchi game (by Theorem~\ref{thm:reduction}). Whereas solving the latter games is $\mathsf{coNP}$-complete by Theorem~\ref{thm:gameresults}, solving the \tproblem\ for ordered B\"uchi games is only polynomial when they have a compact embedding (see Theorem~\ref{thm:poly}).

There is no hope to extend Theorem~\ref{thm:poly} to the other $\omega$-regular objectives studied in this article, unless $\mathsf{P} = \mathsf{PSPACE}$. Indeed, we have $\mathsf{PSPACE}$-hardness of the threshold problem for the following lexicographic games.

\begin{theorem}\label{thm:homogeneoushyp} (1) Lexicographic games have a compact embedding and 
(2) the threshold problem is $\mathsf{PSPACE}$-hard for lexicographic reachability, safety, Rabin, Streett, parity, and Muller games.
\end{theorem}

Note that  we obtain the following corollary as a direct consequence of Theorems~\ref{thm:poly} and~\ref{thm:homogeneoushyp}.

\begin{corollary}\label{cor:lexiproblem}
The threshold problem for lexicographic B\"uchi, co-B\"uchi and explicit Muller games is $\mathsf{P}$-complete and is $\mathsf{PSPACE}$-complete for lexicographic safety, reachability, parity, Streett, Rabin and Muller games.
\end{corollary}
\begin{proof}
The $\mathsf{P}$-membership for lexicographic B\"uchi, co-B\"uchi, and explicit Muller games games follows from Theorem~\ref{thm:poly} and Part (1) of Theorem~\ref{thm:homogeneoushyp}, while $\mathsf{P}$-hardness follows from Theorem~\ref{thm:gameresults} as with $n =1$ lexicographic B\"uchi, co-B\"uchi and explicit Muller games are respectively (classical) B\"uchi, co-B\"uchi, and explicit Muller games. For the other $\omega$-regular objectives, $\mathsf{PSPACE}$-hardness follows from Part (2) of Theorem~\ref{thm:homogeneoushyp}, while $\mathsf{PSPACE}$-membership follows from the work of Bouyer et al.~\cite{BouyerBMU12}.
\qed \end{proof}

The rest of this section is devoted to the proof of Theorem~\ref{thm:homogeneoushyp}. 

\paragraph{\bf Lexicographic games.~}
We now focus on the lexicographic preorder $\precsim$. 
Let us first provide several useful terminology and comments on this preorder. Recall that the lexicographic preorder is monotonic. It is also total, hence $x \sim y$ if and only if $x = y$, and $x \prec y$ if and only if $\neg (y \precsim x)$.
Given a vector $x \in \{0,1\}^n$, we denote by $\overline{x}$ the \emph{complement} of $x$, i.e. $\overline{x}_i = 1- x_i$, for all $i \in \{1, \ldots, n\}$. We denote by $x - 1$ the \emph{predecessor} of $x \ne 0^n$, that is, the greatest vector which is strictly smaller than $x$. We define the \emph{successor} $x+1$ of $x$ similarly. In the sequel, as the \tproblem\ is trivial for $x = 0^n$, we do not consider this threshold. By abuse of notation, we keep writing $x \in \{0,1\}^n$ without mentioning that $x \ne 0^n$. We denote by $\last(x)$ the last index $i$ of $x$ such that $x_i=1$, i.e. $\last(x) = \max\{i \in \{1,\ldots,n\} \mid x_i = 1 \}$. Note that $\playerOne$ can ensure a payoff $\succsim x \ne 0^n$ if and only if he can ensure a payoff $\succ x-1$, and when $\playerTwo$ can avoid a payoff $\succsim x$, we rather say that $\playerTwo$ can \emph{ensure} a payoff $\prec x$. 


We now prove that the lexicographic games have a compact embedding (Part (1) of Theorem~\ref{thm:homogeneoushyp}): we first show that $\s$ is polynomial in Proposition~\ref{prop:size}, and we then show that $\sprim$ is also polynomial in Proposition~\ref{prop:hypo2}.

\begin{proposition}\label{prop:size}
Let $x \in \{0,1\}^n$. Then the set $\M(x)$ is equal to $\{x\} \cup \{y^j \in \{0,1\}^n \mid x_j = 0 \mbox{ $\wedge$ } j < \last(x) \}$, where for all $j \in \{1,\ldots,\last(x)-1\}$, we define the vector $y^j \in \{0,1\}^n$ as equal to $x_1 \ldots x_{j-1} 1 0^{n-j}$ ($x$ and $y^j$ share the same (possibly empty) prefix $x_1 \ldots x_{j-1}$). Moreover, $\s = |\M(x)| \leq n$.
\end{proposition}

\begin{example}
Consider the vector $x = 0010100$ such that $\last(x) = 5$. Then, the set $\M(x)$ is equal to $\{x\} \cup \{1000000, 0100000, 0011000\}$.
\end{example}

\begin{proof}[of Proposition~\ref{prop:size}] We recall that $\M(x)$ is the set of minimal elements (with respect to the subset preorder $\subseteq$) of the set of payoffs $y \succsim x$ embedded in the set $\{0,1\}^n$ ordered with $\subseteq$. Let us show both inclusions between $\M(x)$ and $M = \{x\} \cup \{y^j \in \{0,1\}^n \mid x_j = 0 \mbox{ $\wedge$ } j < \last(x) \}$. 

Let $y \in \M(x)$. If $y = x$, then trivially $y \in M$. Otherwise, assume $y \succ x$ and let $j$ be the first index such that $y_j = 1$ and $x_j = 0$. Note that $x_1 \ldots x_{j-1} = y_1 \ldots y_{j-1}$ since $y \succ x$. We associate with $y$ the vector $y^j =  y_1 \ldots y_{j-1} 1 0^{n-j}$. Note that $y^j \succ x$. By minimality of $y$ and by construction of $y^j$, we obtain $y = y^j$ showing that $y \in M$. 

For the second inclusion, as the lexicographic preorder is monotonic, we have $x \in \M(x)$. Now, consider some $y^j \in M$ such that $x_j = 0$ and $j < \last(x)$. Let us show that $y^j$ belongs to $\M(x)$, that is, $y^j \succsim x$ and there is no $y \succsim x$, $y \ne y^j$, such that $y \subset y^j$ (i.e. $\{i \mid y_i =1 \} \subset \{ i \mid y^j_i = 1\}$). First, we clearly have $y^j \succsim x$ since $y^j = x_1 \ldots x_{j-1}10^{n-j}$ and $x_j =0$. Towards a contradiction, assume now that there exists some $y \succsim x$, $y \ne y^j$, such that $y \subset y^j$. Let $i$ be the first index such that $y_i=0$ and $y^j_i =1$. As $y \subset y^j$, we have $i \leq j$. If $i < j$, then $y$ has $x_1 \ldots x_{i-1}0$ as prefix, $y_i^j = x_i = 1$, showing that $y \prec x$ in contradiction with $y \succsim x$. If $i = j$, then $y = x_1 \ldots x_{j-1}0^{n-j+1}$, and again $y \prec x$ since $j < \last(x)$ by construction of $y^j$.
\qed
\end{proof}

In order to show that $\sprim$ is polynomial in Proposition~\ref{prop:hypo2}, we need to proof the following proposition that establishes the link between the duality between a payoff defined with some objectives and the payoff defined with the opposite objectives.

\begin{proposition}\label{thm:dual}
Let $(G,\Obj_1,\ldots, \Obj_n,\precsim)$ be a lexicographic game, $\mu \in \{0,1\}^n$ be a threshold, and $v_0$ be an initial vertex. Then $\playerOne$ can ensure a payoff $\succeq \mu$ in the lexicographic game $(G,\Obj_1,\ldots, \Obj_n,\precsim)$ if and only if $\playerOne$ can ensure a payoff $\lexi \overline{\mu}$ in the lexicographic game $(G,\overline{\Obj_1},\ldots,\overline{\Obj_n},\precsim)$.
\end{proposition}

\begin{proof}
Recall that $\playerOne$ can ensure a payoff $\succeq \mu$ from $v_0$ in the lexicographic game $(G,\Obj_1,\ldots, \Obj_n,\precsim)$ if and only if he has a winning strategy from $v_0$ for the objective $\{\pi \mid \payoff(\pi) = \nu \succeq \mu\}$. Moreover, for any $i$, $\nu_i = 1$ if and only if $\pi \in \Obj_i$, i.e. $\pi \not\in \overline{\Obj_i}$, and $\nu_i=0$ if and only if $\pi \not\in \Obj_i$, i.e. $\pi \in \overline{\Obj_i}$. 
Thus, $\overline{\nu}_i = 0$ iff $\pi \not\in \overline{\Obj_i}$ and $\overline{\nu}_i =1$ iff $\pi \in \overline{\Obj_i}$. 
Then, as $\nu \succeq \mu$ iff $\overline{\nu} \preceq \overline{\mu}$, we have that $\playerOne$ is winning  from $v_0$ for the objective $\{\pi \mid \payoff(\pi) = \nu \succeq \mu\}$ if and only if $\playerOne$ has a strategy to ensure a payoff $\preceq \overline{\mu}$ from $v_0$ in the lexicographic game $(G,\overline{\Obj_1},\ldots,\overline{\Obj_n},\precsim)$.
\qed
\end{proof}

\begin{proposition}\label{prop:hypo2}
Let $(G,\Obj_1,\ldots, \Obj_n,\precsim)$ be a lexicographic B\"uchi game and $\mu \in \{0,1\}^n$. Then, the objective $\Obj = \cup_{\nu \in \M(\mu)} \cap_{i \in \delta_{\nu}} \Obj_i$ can be rewritten in conjunctive normal form with a conjunction of size $\sprim \leq n$.
\end{proposition}

\begin{proof}
By Proposition~\ref{thm:dual} and Item~\ref{item:safety} of Proposition~\ref{rem:objectives}, $\playerOne$ can ensure a payoff $\succsim \mu$ in $(G,\Obj_1,\ldots, \Obj_n,\precsim)$ if and only if $\playerOne$ can ensure a payoff $\lexi \overline{\mu}$ in the lexicographic co-B\"uchi game $(G,\overline{\Obj_1},\ldots,\overline{\Obj_n},\precsim)$. By Thereom~\ref{thm:reduction} and Corollary~\ref{cor:deter}, equivalently, $\playerTwo$ cannot satisfy the objective $\cup_{\nu \in \M(\overline \mu + 1)} \cap_{i \in \delta_{\nu}} \overline{\Obj_i}$. This is equivalent to say that $\playerOne$ can satisfy the complement of the latter objective, that is, the objective $\cap_{\nu \in \M(\overline \mu + 1)} \cup_{i \in \delta_{\nu}} \Obj_i$. We have $|\M(\overline \mu + 1)| \leq n$ by Proposition~\ref{prop:size}.
\qed\end{proof}

We finally prove Part (2) of Theorem~\ref{thm:homogeneoushyp}.

\begin{proof}[of Theorem~\ref{thm:homogeneoushyp}, Part (2)]
Let us study the complexity lower bounds.
\begin{itemize}
\item The $\mathsf{PSPACE}$-hardness of the \tproblem\ for lexicographic reachability (resp. safety) games is obtained thanks to a polynomial reduction from solving generalized reachability games  which is $\mathsf{PSPACE}$-complete by Theorem~\ref{thm:gameresults}. Let $(G,\Obj)$ be a generalized reachability game with $\Obj = \GenReach(U_1,\ldots,U_n)$. Let $(G,\Obj_1,\ldots, \Obj_n,\precsim)$ be the lexicographic reachability (resp. safety) game with $\Obj_i = \Reach(U_i)$ (resp. $\Obj_i = \Safe(U_i^c)$)~$\forall i$.
\begin{itemize}
\item Reachability: We have that $\playerOne$ is winning in $(G,\Obj)$ from $v_0$ if and only if $\playerOne$ can ensure a payoff $\succsim \mu = 1^n$ from $v_0$ in the lexicographic reachability game $(G,\Obj_1,\ldots, \Obj_n,\precsim)$.
\item Safety: We claim that $\playerOne$ is winning in $(G,\Obj)$ from $v_0$ if and only if $\playerOne$ can ensure a payoff $\succsim \mu = 0^{n-1}1$ from $v_0$ in the lexicographic safety game. This follows from the determinacy of generalized reachability games, and from the fact that  $\playerOne$ can ensure a payoff $\succsim \mu$ from $v_0$ in the lexicographic safety game if and only if $\playerTwo$ is losing in the generalized reachability game $(G,\Obj)$ from $v_0$.
\end{itemize}
\item The hardness of the \tproblem\ for lexicographic parity games is obtained thanks to a polynomial reduction from solving games $(G,\Obj)$ the objective $\Obj$ of which is a union of a Rabin objective and a Streett objective, which is known to be $\mathsf{PSPACE}$-complete~\cite{AlurTM03}. Let $\Obj = \Rabin((E_i,F_i)_{i=1}^{n_1}) \cup \Streett((E_i,F_i)_{i=n_1+1}^{n})$. As any Rabin (resp. Streett) objective is the union (resp. intersection) of parity objectives by Item~\ref{item:Streett} of Proposition~\ref{rem:objectives}, we can rewrite $\Obj$ as $\Obj = \cup_{i=1}^{n_1}(\Par(p_i)) \cup (\cap_{i=n_1+1}^n \Par(p_i))$,
where all $p_i$ are coloring functions. Let $(G,\Obj_1,\ldots, \Obj_n,\precsim)$ be the lexicographic parity game where $\Obj_i = \Par(p_i)$ for all $i$. We claim that $\playerOne$ is winning in the game $(G,\Obj)$ from $v_0$ if and only if $\playerOne$ can ensure a payoff $\succsim \mu$ from $v_0$ in the lexicographic parity game $(G,\Obj_1,\ldots, \Obj_n,\precsim)$ where  $\mu = 0^{n_1}1^{n-n_1}$.
Indeed, if a play $\pi$ satisfies $\payoff(\pi) \succsim \mu$ then either $\payoff(\pi) = \mu$ in which case $\pi \in \cap_{i=n_1+1}^n \Par(p_i)$, i.e. $\pi$ satisfies the Streett objective, or $\payoff(\pi) \succ \mu$ in which case there exists $1 \in \{1,\ldots, n_1 \}$ such that $\pi \in \Par(p_i)$, i.e. $\pi$ satisfies the Rabin objective. Conversely, if a play $\pi$ satisfies the Streett or the Rabin objective then $\payoff(\pi) \succsim \mu$ since $\payoff(\pi) \succsim \mu$ (resp. $\succ \mu$) as soon as $\pi$ satisfies the Streett (resp. Rabin) objective. 
\item As parity objectives are a special case of Rabin (Streett) objectives by Item~\ref{item:parityRabin} of Proposition~\ref{rem:objectives}, the lower bound follows (from the previous item) for both lexicographic Rabin and Streett games. 
\item Lexicographic Muller games with $n = 1$ and $\mu = 1$ are a special case of Muller games and solving the latter games is PSPACE-complete by Theorem~\ref{thm:gameresults}.
\end{itemize}

\noindent
This completes the proof. \qed\end{proof}
\section{Values and optimal strategies in lexicographic games}\label{sec:values}

In this section, we first recall the notion of values and optimal strategies. We then show how to compute the values in lexicographic games, and what are the memory requirements for the related optimal strategies. This yields a \textit{full picture} of the study of lexicographic games, see Table~\ref{table:homogeneous}. In this table, the first row indicates the complexity of the \tproblem\ (Corollary~\ref{cor:lexiproblem}) and the remaining rows summarize the results on the values and the optimal strategies (Theorem~\ref{thm:valueAndMemory} herafter).

\begin{table}
\begin{center}
\scalebox{0.80}{
\begin{tabular}{|c||c|c|c|c|c|c|c|c|c|}
\cline{2-10}
\multicolumn{1}{c|}{} & ~Reachability~ & ~Safety~ & ~B\"uchi~ & ~Co-B\"uchi~ & ~Explicit Muller~ & ~Parity~ & ~Rabin~ & ~Streett~ & ~Muller~ \\
\cline{1-10}
 ~Threshold problem~ & \multicolumn{2}{c|}{~$\mathsf{PSPACE}$-complete~} & \multicolumn{3}{c|}{$\mathsf{P}$-complete} & \multicolumn{4}{c|}{$\mathsf{PSPACE}$-complete} \\
\hline
Values & \multicolumn{2}{c|}{exponential and FPT} & \multicolumn{3}{c|}{polynomial} & \multicolumn{4}{c|}{exponential and FPT}\\
\hline
~$\playerOne$ memory~ & \multicolumn{2}{c|}{\multirow{2}{*}{exponential}} & linear & ~memoryless~ & \multicolumn{5}{c|}{\multirow{2}{*}{exponential}} \\
\cline{0-0} \cline{4-5}
~$\playerTwo$ memory~ & \multicolumn{2}{c|}{} & ~memoryless~ & linear & \multicolumn{5}{c|}{} \\
\hline
\end{tabular}}
\end{center}
\caption{Overview of the results on lexicographic games with $\omega$-regular objectives. The second row indicates the complexity time of computing the values. The third and last rows indicate the tight memory requirements of the winning and optimal strategies for both players.}
\label{table:homogeneous}
\end{table}

\paragraph{\bf Values and optimal strategies.~}

In a lexicographic game, one can define the best reward that $\playerOne$ can ensure from a given vertex, that is, the highest threshold $\mu$ for which $\playerOne$ can ensure a payoff $\succsim \mu$. Dually, we can also define the worth reward that $\playerTwo$ can ensure. More precisely, if there exists some $\mu \in \{0,1\}^n$ and two strategies $\sigma_1^* \in \Sigma_1, \sigma_2^* \in \Sigma_2$ such that $\payoff (\Out(v,\sigma_1,\sigma_2^*)) \precsim \mu \precsim \payoff (\Out(v,\sigma_1^*,\sigma_2))$ for all strategies $\sigma_1 \in \Sigma_1, \sigma_2 \in \Sigma_2$, then $\mu$ is called the \emph{value} $\mathsf{Val}(v)$ of $v$ and $\sigma_1^*, \sigma_2^*$ are called \emph{optimal} strategies from $v$. Note that the play $\pi = \Out(v,\sigma_1^*,\sigma_2^*)$ consistent with both optimal strategies has payoff $\payoff(\pi) = \mathsf{Val}(v)$. 
The lexicographic game $(G,\Obj_1,\ldots,\Obj_n,\precsim)$ is called \emph{value-determined} if $\mathsf{Val}(v)$ exists for every $v \in V$.

%
%

We have the following theorem for lexicographic games. The rest of the section is devoted to its proof.

\begin{theorem}\label{thm:valueAndMemory}
(1) The value of each vertex in lexicographic B\" uchi, co-B\" uchi, and explicit Muller games can be computed with a polynomial time algorithm, and with an exponential time and an FPT algorithm for lexicographic reachability, safety, parity, Rabin, Streett, and Muller games.

(2) The following assertions hold for both winning strategies of the \tproblem\ and optimal strategies. Linear memory strategies are necessary and sufficient for $\playerOne$ (resp. $\playerTwo$) while memoryless strategies are sufficient for $\playerTwo$ (resp. $\playerOne$) in lexicographic B\"uchi (resp. co-B\"uchi) games. Exponential memory strategies are both necessary and sufficient for both players in lexicographic reachability, safety, explicit Muller, parity, Rabin, Streett, and Muller games.
\end{theorem}

\paragraph{\bf Proof of Part (1) of Theorem~\ref{thm:valueAndMemory}.~} When the objectives are Borel sets, the following proposition states the value-determinacy of lexicographic games. It also states that an algorithm for the threshold problem leads to an algorithm for computing the values with a complexity multiplied by $n$. 
Hence the first part of Theorem~\ref{thm:valueAndMemory} immediately follows from Corollary~\ref{cor:lexiproblem} and Proposition~\ref{prop:value}.

\begin{proposition}\label{prop:value}
Let $(G,\Obj_1,\ldots,\Obj_n)$ be a lexicographic game. 
If $\Obj_1, \ldots, \Obj_n$ are Borel sets, then the lexicographic game is value-determined. 
Moreover, if the \tproblem\  can be solved with an algorithm of complexity $\C$, then for all $v \in V$, the value $\mathsf{Val}(v)$ can be computed with an algorithm of complexity $n \cdot \C$.
\end{proposition}

\begin{proof}
Let us show first that the lexicographic game $(G,\Obj_1,\ldots,\Obj_n)$ is value-determined. To this end, we use the following Folk property: if there exists some $\alpha \in \{0,1\}^n$ and two strategies $\sigma_1^* \in \Sigma_1, \sigma_2^* \in \Sigma_2$ such that $\payoff (\Out(v,\sigma_1,\sigma_2^*)) \precsim \alpha \precsim \payoff (\Out(v,\sigma_1^*,\sigma_2))$ for all strategies $\sigma_1 \in \Sigma_1, \sigma_2 \in \Sigma_2$, then $\alpha =  \mathsf{Val}(v)$ and $\sigma_1^*, \sigma_2^*$ are optimal strategies from $v$. Let $v$ be a vertex. The set of thresholds $\mu$ is partitioned between the two players according to whether $\playerOne$ (resp. $\playerTwo$) can ensure a payoff $\succsim \mu$ (resp. $\prec \mu+1$) from $v$ by Corollary~\ref{cor:deter}. Let $\alpha$ be the highest threshold that $\playerOne$ can ensure and $\sigma_1^*$ be the corresponding winning strategy. By definition of $\alpha$, $\playerTwo$ can ensure a payoff $\prec \mu +1$ with a winning strategy $\sigma_2^*$. It follows that $\payoff (\Out(v,\sigma_1,\sigma_2^*)) \precsim \alpha \precsim \payoff (\Out(v,\sigma_1^*,\sigma_2))$ for all strategies $\sigma_1 \in \Sigma_1, \sigma_2 \in \Sigma_2$, and therefore we have $\mathsf{Val}(v) = \alpha$.

When the \tproblem\ is decidable, the procedure to compute the value $\mathsf{Val}(v)$ works as follows. The idea is to solve the threshold problem for different thresholds from vertex $v$ in a way to compute the highest threshold $\mu$ for which $\playerOne$ can ensure a payoff $\succsim \mu$. This threshold $\mu$ is the value $\mathsf{Val}(v)$.

First, we test whether $\playerOne$ can ensure a payoff $\succsim 10^{n-1}$ from $v$. If this is the case, we set bit $\mu_1$ to $1$ and  to $0$ otherwise. Then, for $i \in \{2,\ldots,n\}$, we successively test whether $\playerOne$ can ensure a payoff $\succsim \mu_1 \ldots \mu_{i-1} 1 0^{n-i}$ from $v$ and we set bit $\mu_i$ to $1$ if this is the case and to $0$ otherwise.
Thus, after those $n$ solutions to the \tproblem, we obtain a threshold $\mu = \mu_1 \ldots \mu_n$ for which $\playerOne$ can ensure a payoff $\succeq \mu$ from $v$. The complexity of the algorithm computing $\mu$ is thus in $n \cdot \C$. By using again the previous Folk property with the computed $\mu$, we have that $\mu = \mathsf{Val}(v)$.
 
This concludes the proof. \qed
\end{proof}

\begin{remark}\label{rem:valeur} 
When the procedure given in the proof of Proposition~\ref{prop:value} computes the value $\mu$ of a given vertex $v$, it also computes at the same time an optimal strategy from this vertex for both players. Indeed, the optimal strategy $\sigma_1^*$ of $\playerOne$ from $v$ is his winning strategy (for the \tproblem) that ensures a payoff $\succsim \mu$ and that the optimal strategy $\sigma_2^*$ of $\playerTwo$ from $v$ is his winning strategy that ensures a payoff $\prec \mu+1$. Notice that in this procedure $\sigma_1^*$ (resp. $\sigma_2^*$) is the winning stategy of $\playerOne$ (resp. $\playerTwo$) for the last bit $\mu_i$ set to $1$ (resp. to $0$).
Therefore, in order to study some properties on optimal strategies (such as memory requirements), it is sufficient to study winning strategies for the \tproblem.
\end{remark}

\begin{example}\label{ex:calculValeur}

Let us consider the lexicographic reachability game $(G,\Obj_1,\Obj_2,\Obj_3,\precsim)$ depicted on Figure~\ref{fig:calculValeur} where $\Obj_1 = \Reach(\{v_1\})$, $\Obj_2 = \Reach(\{v_2,v_4\})$ and $\Obj_3 = \Reach(\{v_5\})$ and $\precsim$ is the lexicographic order. We apply the procedure described in Proposition~\ref{prop:value} to compute $\mathsf{Val}(v_0)$ and the corresponding optimal strategies. For this purpose, we begin by testing whether $\playerOne$ can ensure a payoff $\succsim 100$ from $v_0$. This is not the case since $\playerTwo$ can prevent him from visiting vertex $v_1$ by going from $v_0$ to $v_2$. In particular, this strategy of $\playerTwo$ ensures a payoff $\prec 100$. We fix $\mu_1 = 0$ and we now test whether $\playerOne$ can ensure a payoff $\succsim 010$. This is the case since by going from $v_3$ to $v_4$, $\playerOne$ is guaranteed to visit vertex $v_4$. We thus set $\mu_2 = 1$. The final test made is whether $\playerOne$ can ensure a payoff $\succsim 011$. This is possible with the strategy that consists in going from $v_3$ to $v_5$. Indeed, the two possible outcomes consistent with this strategy are $v_0v_1v_3v_5^{\omega}$ and $v_0v_2v_3v_5^{\omega}$. The payoff of the first play is $101 \succsim 011$ while the latter payoff is $011$. Hence we set $\mu_3=1$ and we get $\mathsf{Val(v_0)} = 011$. The corresponding optimal strategies are to go from $v_3$ to $v_5$ for $\playerOne$ (to ensure a payoff $\succsim 011$) and to go from $v_0$ to $v_2$ for $\playerTwo$ (to ensure a payoff $\prec 100$). Note that the outcome from $v_0$ consistent with those strategies is $v_0v_2v_3v_5^{\omega}$ and that its payoff is indeed $\mathsf{Val}(v_0)$.

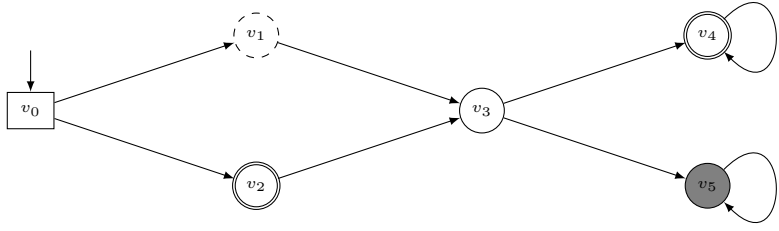
\begin{figure}[h]
\centering
  \begin{tikzpicture}[scale=4]
    \everymath{\scriptstyle}
    \draw (0,0) node [rectangle, inner sep=5pt, draw] (A) {$v_0$};
    \draw (0.75,0.25) node [circle, draw, dashed] (B) {$v_1$};
    \draw (0.75,-0.25) node [circle, draw, double] (C) {$v_2$};
	\draw (1.5,0) node [circle, draw] (D) {$v_3$};
    \draw (2.25,0.25) node [circle, draw, double] (E) {$v_4$};
    \draw (2.25,-0.25) node [circle, draw, fill=gray] (F) {$v_5$};
    
    \draw[->,>=latex] (A) to (B);
    \draw[->,>=latex] (A) to (C);
    \draw[->,>=latex] (C) to (D);
    \draw[->,>=latex] (B) to (D);
    \draw[->,>=latex] (D) to (E);
    \draw[->,>=latex] (D) to (F);     
    
    \draw[->,>=latex] (E) .. controls +(45:0.4cm) and +(315:0.4cm) .. (E);
    \draw[->,>=latex] (F) .. controls +(45:0.4cm) and +(315:0.4cm) .. (F);
	\path (0,0.2) edge [->,>=latex] (A);    
    
    \end{tikzpicture}
\caption{The value of vertex $v_0$ is equal to 011.}
\label{fig:calculValeur}
\end{figure}
\end{example}

\paragraph{\bf Proof of Part (2) of Theorem~\ref{thm:valueAndMemory}.~}
Thanks to Remark~\ref{rem:valeur}, in order to prove the second part of Theorem~\ref{thm:valueAndMemory}, it is sufficient to study the memory requirements of winning strategies as those of optimal strategies are the same. Upper bounds on the memory sizes are obtained by analyzing the various reductions done in the proof of Theorem~\ref{thm:lexiFPT} in the case of a preorder with a compact embedding. Lower bounds for lexicographic B\"uchi and co-B\"uchi games are obtained thanks to a reduction from generalized B\"uchi games, and for the other lexicographic games thanks to the reductions proposed in the proof of Part (2) of Theorem~\ref{thm:homogeneoushyp}.

\begin{proof}[of Part (2) of Theorem~\ref{thm:valueAndMemory}]
By Remark~\ref{rem:valeur}, we only study the memory requirements of winning strategies for the threshold problem. The proof is split into two parts dealing first with the upper bounds and then with the lower bounds.

Concerning the upper bounds, we first recall that for any ordered game with a compact embedding, the proof of Theorem~\ref{thm:lexiFPT} yields a reduction of
\begin{itemize}
\item lexicographic reachability (resp. safety) games to \unioninter\ reachability (resp. \unioninter\ safety) games,
\item lexicographic B\"uchi (resp. co-B\"uchi) games to (resp. the complement of) generalized B\"uchi games,
\item lexicographic explicit Muller games to explicit Muller games and
\item lexicographic parity, Rabin, Streett, and Muller games to Boolean B\"uchi games.
\end{itemize}
Those reductions do not modify the initial game structure and winning strategies for the games obtained by the reductions are winning strategies for the \tproblem\ of the original lexicographic games. As linear memory (resp. memoryless) strategies are sufficient for $\playerOne$ (resp. $\playerTwo)$ in generalized B\"uchi games and exponential memory strategies are sufficient for both players in Explicit Muller, Boolean B\"uchi, \unioninter\ reachability, and \unioninter\ safety games by Theorem~\ref{thm:gameresults}, we obtain the expected upper bounds.

Concerning the lower bounds, we rely on the reductions proposed in the proof of Part (2) of Theorem~\ref{thm:homogeneoushyp}.
\begin{itemize}
\item As there is a reduction from solving generalized reachability (resp. explicit Muller, Muller) games to solving the \tproblem\ for lexicographic reachability and safety (resp. explicit Muller, Muller) games, exponential memory is necessary for both players by Theorem~\ref{thm:gameresults}.
\item There is a reduction from solving games the objective of which is a union of a Rabin and a Streett objective to the threshold problem for lexicographic parity games. Thus, the latter problem is harder than solving both Rabin and Streett games, which implies that exponential memory is necessary for both players in lexicographic parity games by Theorem~\ref{thm:gameresults}. This is also the case for lexicographic Rabin and Streett games, since parity objectives are a special case of Rabin (Streett) objectives by Item~\ref{item:parityRabin} of Proposition~\ref{rem:objectives}.
\item It remains to show that $\playerOne$ (resp. $\playerTwo$) needs linear memory in lexicographic B\"uchi (resp. co-B\"uchi) games. This is obtained thanks to Theorem~\ref{thm:gameresults} and the following reductions from generalized B\"uchi games. Let $(G,\Obj)$ with $\Obj = \GenBuchi(U_1,\ldots,U_n)$
\begin{itemize}
\item B\"uchi case: We have that  $\playerOne$ is winning in $(G,\Obj)$ from $v_0$ if and only if $\playerOne$ can ensure a payoff $\succeq 1^n$ from $v_0$ in the lexicographic B\"uchi game $(G,\Buchi(U_1),\ldots,\Buchi(U_n),\precsim)$. 
\item Co-B\"uchi case: Note that any play $\pi$ belongs to $\GenBuchi(U_1,\ldots,U_n)$ if and only if $\pi \not\in \CoBuchi(U_i^c)$ for all $i$, i.e. $\payoff(\pi) = 0^n$. Hence, $\playerOne$ is winning in $(G,\Obj)$ from $v_0$ if and only if, taking on the role of $\playerTwo$, he can ensure a payoff $\prec \mu = 0^{n-1}1$ in the lexicographic co-B\"uchi game $(G,\CoBuchi(U_1^c),\ldots,$ $\CoBuchi(U_n^c),\precsim)$.
\end{itemize}
\end{itemize}
This finishes the proof. \qed
\end{proof}
\section{Conclusion}

In this paper, we have studied the parameterized complexity of the threshold problem for monotonically ordered games with $\omega$-regular objectives when the set of objectives is taken as a parameter. This latter result is particularly relevant as in practice, the number of objectives is usually restricted. We have also studied the special case of the lexicographic order, and given a full ficture of the study of lexicographic games. In particular, we have shown how to compute the values in lexicographic games.

As future work, we would like to investigate notions of equilibria for those games, as well as subgame perfection: a strategy is subgame perfect if it ensures the maximal value that is achievable in every subgame. Also, we would like to study lexicographic games with quantitative objectives.

\medskip\noindent
{\bf Acknowledgments.~} We would like to thank Antonia Lechner for useful discussions.

\bibliographystyle{abbrv}
\bibliography{biblio}

\end{document}